\documentclass[british]{scrartcl}
\usepackage[T1]{fontenc}
\usepackage[latin9]{inputenc}
\usepackage{float}
\usepackage{bm}
\usepackage{amsmath}
\usepackage{accents}
\usepackage{amssymb}
\usepackage{graphicx}
\usepackage{esint}
\usepackage[authoryear]{natbib}
\usepackage{amsthm}
\usepackage{url}
\usepackage{subcaption}
\usepackage{enumitem}

\usepackage{todonotes}

\makeatletter
\providecommand\theHALG@line{\thealgorithm.\arabic{ALG@line}}
\makeatother

\makeatletter

\floatstyle{ruled}
\newfloat{algorithm}{tbp}{loa}
\providecommand{\algorithmname}{Algorithm}
\floatname{algorithm}{\protect\algorithmname}

\usepackage{dsfont}\usepackage{algpseudocode}
\usepackage{mathtools}
\usepackage{color}

\newcounter{rmq}[section]
\newcommand{\setX}{\mathcal{X}}

\newcommand{\ui}{[0,1)} 

\renewcommand{\P}{\mathbb{P}}
\newcommand{\E}{\mathbb{E}}

\newcommand{\Prob}{\mathrm{Pr}}

\newcommand{\F}{\mathcal{F}} 

\newcommand{\Unif}{\mathcal{U}} 

\newcommand{\bx}{\mathbf{x}}

\newcommand{\bu}{\mathbf{u}}
\newcommand{\bbu}{\mathbf{U}}
\newcommand{\bv}{\mathbf{v}}
\newcommand{\by}{\mathbf{y}}
\newcommand{\bz}{\mathbf{z}}

\newcommand{\ind}{\mathds{1}}
\newcommand{\dd}{\mathrm{d}}
\newcommand{\dx}{\dd \bx}

\newcommand{\bigO}{\mathcal{O}} 
\newcommand{\cvz}{\rightarrow 0} 
\renewcommand{\emptyset}{\varnothing} 

\newcommand{\comment}[1]{ \ifthenelse{ \equal{\showcomment}{true} }{ {\bf #1} }{} }
\newcommand{\showcomment}{true}
\newcommand{\subscript}[2]{$#1 _ #2$}

\newtheorem{theorem}{Theorem}

\newtheorem{corollary}{Corollary}

\newtheorem{lemma}{Lemma}
\newtheorem{definition}{Definition}

\makeatother

\usepackage{babel}

\begin{document}

\title{Improving Simulated Annealing through Derandomization}
\author{Mathieu Gerber\thanks{corresponding author: mathieugerber@fas.harvard.edu} 
\and Luke Bornn}
\date{Department of Statistics\\ \vspace{0.3cm}Harvard University}
\maketitle

\begin{abstract}
We propose and study a version of simulated annealing (SA) on continuous state spaces based on $(t,s)_R$-sequences. The parameter $R\in\bar{\mathbb{N}}$  regulates the degree of randomness of the input sequence, with the case $R=0$ corresponding to IID uniform random numbers  and the limiting case $R=\infty$ to   $(t,s)$-sequences. Our main result, obtained for rectangular domains, shows that the resulting optimization  method, which we refer to as QMC-SA, converges almost surely to the global optimum of the objective function $\varphi$ for any $R\in\mathbb{N}$. When  $\varphi$ is univariate, we are in addition  able to show that the completely deterministic version of QMC-SA is convergent. A key property of these  results is that they do not require objective-dependent conditions on the cooling schedule. As a corollary of our theoretical analysis, we provide a new almost sure convergence result for SA which shares this property under minimal assumptions on $\varphi$. We further explain how our results in fact apply to a broader class of optimization methods including for example threshold accepting, for which to our knowledge no convergence results currently exist. We finally illustrate the superiority of QMC-SA over  SA algorithms in a numerical study.

\end{abstract}

\section{Introduction}

Simulated annealing (SA)  belongs to the class of stochastic optimization techniques, a popular suite of tools to find the global optimum of multi-modal and/or non-differentiable functions defined on continuous sets \citep{Geman1986}. For instance, SA has been successfully used to solve complicated continuous optimization problems arising in  signal processing \citep{Chen1999}, antenna array synthesis \citep{Girard2001}  and thermodynamics \citep{Zhang2011}; see \citet{Locatelli2002} and references therein for more applications of SA on continuous state spaces. In addition, since many inferential problems amount to optimizing an objective function, SA has proved to be useful both with frequentist \citep[see, e.g.,][]{Goffe1994, Rubenthaler2009,  Chen2011maximum} and Bayesian \citep[see, e.g.,][]{Andrieu2000, Ireland2007} statistical methods.

In this paper, we propose and study  SA algorithms based on quasi-Monte Carlo (QMC) point sets -- informally, sets of points that more evenly cover the unit hypercube than uniform random numbers. QMC point sets are already widely used in numerical analysis to evaluate the integral of a function over a hypercube, and an important literature on QMC integration error rate has been developed during the last 30 years  \citep[][]{dick2010digital}. Since QMC point sets are designed to evenly cover hypercubes, they intuitively should be an efficient alternative to pseudo-random numbers inside stochastic optimization routines. However, their use for global optimization is surprisingly scarce and is mainly limited to quasi-random search and some ad-hoc improvements thereof; see Section \ref{sub:QMCopt} for a   literature review.

More precisely,  we develop a SA algorithm (hereafter referred to as QMC-SA)  using as input $(t,s)_R$-sequences. The parameter $R\in\bar{\mathbb{N}}$ regulates the degree of randomness of the sequence, with the case $R=0$ corresponding to IID uniform random numbers in $\ui^s$ and the limiting case $R=\infty$ to   $(t,s)$-sequences, which  encompass the most classical constructions of QMC sequences such as Sobol', Faure and Niederreiter-Xing sequences \citep[see, e.g,][Chapter 8]{dick2010digital}.  The case $R=\infty$ (i.e. deterministic QMC-SA) is of particular interest but our main result only holds for $R\in\mathbb{N}$ because, for multivariate objective functions, some  randomness is needed to rule out  odd behaviours that may be hard to exclude with deterministic point sets. Note that our convergence result is valid for \emph{any} $R\in\mathbb{N}$. Consequently, only arbitrary small randomness is needed  and thus, as explained below (Section \ref{sub:(t-s)_R}), can be omitted in practice. For univariate test functions, our convergence result  also applies for the limiting case $R=\infty$.

For ``standard'' (i.e. Monte Carlo) SA algorithms, the choice of the cooling schedule $(T_n)_{n\geq 1}$ is a critical one since most convergence results for SA impose conditions on this sequence that are model-dependent and usually intractable. This is for instance the case for the almost sure convergence result of \citet{Belisle1992} on compact spaces or the convergence in probability result of \cite{Andrieu2001} on unbounded domains.

Regarding this point, the theoretical guarantees of QMC-SA, established on rectangular spaces, have the advantage to impose no problem-dependent conditions on the cooling schedule, only requiring one to choose a sequence  $(T_n)_{n\geq 1}$ such that the series $\sum_{n=1}^{\infty}(T_n\log n)$ is convergent. As a corollary of our analysis, we show that, under weaker assumptions than in \citet{Belisle1992}'s result,  selecting such a cooling schedule is enough to ensure the almost sure convergence  of  SA. We also note that our analysis applies for a broader class of optimization methods, which in particular encompass SA with arbitrary acceptance probability function and threshold accepting \citep{dueck1990,Moscato1990}. To the best of our knowledge, no convergence result exists for this latter method and thus this paper provides the first theoretical guarantees for this class of optimization techniques.

From an algorithmic point of view, QMC-SA simply amounts to replacing the pseudo-random random numbers used as input to SA with points taken from a $(t,s)_R$-sequence. The cost of generating the first $N\geq 1$ points of a  $(t,s)_R$-sequence is $\bigO(N\log N)$ for any $R\in\bar{\mathbb{N}}$ and therefore,  for a fixed number of function evaluations,  QMC-SA is slower than classical SA algorithms. However, the extra cost is small because in most cases we can generate points from $(t,s)_R$-sequences using efficient logical operations; further, the cost of simulating these sequences is typically minuscule in comparison to other algorithmic steps, such as evaluating the objective function. In addition, and as illustrated below, the deterministic version of QMC-SA typically requires many fewer function evaluations than SA to reach the same approximation error.

The rest of the paper is organized as follows. In Section \ref{sec:preliminaries} we set the notation and present the background material necessary to follow this work. The QMC-SA algorithm is presented in Section \ref{sec:QMC-SA} and its  convergence study  is given in Section \ref{sec:convergence}. In Section \ref{sec:extension} we discuss some extensions of QMC-SA; in particular, we will see how our convergence results apply to threshold accepting. Section \ref{sec:num} proposes  a numerical study to compare SA and QMC-SA, notably on a non-differentiable and high dimensional  optimization problem that arises in spatial statistics. Section \ref{sec:conc} concludes.

\section{Preliminaries\label{sec:preliminaries}}

In this section we first set the notation we will use throughout this work before giving a brief introduction to SA and to QMC methods. In particular, the use of QMC point sets for global optimization is motivated and discussed in Section \ref{sub:QMCopt}.

\subsection{Notation}

Let  $\setX\subseteq \mathbb{R}^d$ be a measurable set and $\varphi:\setX\rightarrow\mathbb{R}$. Then, we write $\varphi^*=\sup_{\bx\in\setX}\varphi(\bx)$ the quantity of interest in this work and $\mathcal{P}(\setX)$  the set of all probability measures on  $\setX$ which are absolutely continuous with respect to  $\lambda_d(\dx)$, the Lebesgue measure on $\mathbb{R}^d$. Vectors in $\mathbb{R}^d$ are denoted using bold notation and, for two integers $a$ and $b$, $b\geq a$, we write $a:b$ the set of integers $\{a,\dots,b\}$. Similarly, $\bx^{1:N}$ is the set of $N$ points $\{\bx^1,\dots,\bx^N\}$  in $\mathbb{R}^d$ and, for a vector $\bx\in\mathbb{R}^d$, $x_{1:i}:=(x_1,\dots,x_i)$, $i\in 1: d$. For two real numbers $a$ and $b$, we will sometimes use the notation $a\vee b$ (resp. $a\wedge b$) to denote the maximum (resp. the minimum) between $a$ and $b$.

Let $K:\setX\rightarrow \mathcal{P}(\setX)$ be a Markov kernel whose  density with respect to the Lebesgue measure is denoted by $K(\by|\bx$). We write $K_i(\bx,y_{1:i-1}, \dd y_i)$, $i\in 1:d$, the distribution of $y_i$ conditional on $y_{1:i-1}$, relative to $K(\bx,\dd\by)$ (with the convention $K_i(\bx,y_{1:i-1}, \dd y_i)=K_1(\bx, \dd y_1)$ when $i=1$) and the corresponding density  function is denoted by $K_i(y_i|y_{1:i-1},\bx)$ (again with the convention $K_i(y_i|y_{1:i-1},\bx)=K_1(y_1|\bx)$ when $i=1$).

For a distribution $\pi\in\mathcal{P}(\setX)$, we denote by $F_{\pi}:\setX\rightarrow [0,1]^d$ (resp. $F^{-1}_{\pi}:[0,1]^d\rightarrow \setX$) its Rosenblatt (resp. inverse Rosenblatt) transformation \citep{Rosenblatt1952}. When $d=1$, the (inverse)  Rosenblatt transformation of $\pi$ reduces to its (inverse) CDF and, when $d>1$,  the $i$-th component of $F_\pi$ (resp. of $F_\pi^{-1}$) is the CDF (resp. the inverse CDF) of  component $x_i$ conditionally on $(x_1,\dots, x_{i-1})$. Similarly, for a kernel $K:\setX\rightarrow \mathcal{P}(\setX)$, we denote by $F_{K}(\bx,\cdot)$ (resp. $F^{-1}_{K}(\bx,\cdot)$)  the Rosenblatt  (resp. the inverse Rosenblatt) transformation of the probability measure $K(\bx,\dd\by)\in \mathcal{P}(\setX)$. We   refer the reader e.g. to \citet[][Section 3.1]{SQMC} for a more detailed presentation of these mappings.

Finally, we write $B_{\delta}(\tilde{\bx})\subseteq\setX$  the ball of radius $\delta>0$ around $\tilde{\bx}\in\setX$ with respect to the supremum norm, i.e. $B_{\delta}(\tilde{\bx})=\{\bx\in \setX:\, \|\bx-\tilde{\bx}\|_{\infty}\leq \delta\}\cap\setX$ where, for vectors $\bz\in\mathbb{R}^d$, $\|\bz\|_\infty=\max_{i\in 1:d}|z_i|$.

\subsection{Introduction to simulated annealing\label{sub:SA}}

As already mentioned, simulated annealing \citep[see, e.g,][for a recent overview ]{Nikolaev2010}  is an iterative stochastic optimization techniques designed to evaluate the \textit{supremum} of a function $\varphi:\setX\subseteq\mathbb{R}^d\rightarrow \mathbb{R}$; see Algorithm \ref{alg:SAQMC} below for a pseudo-code version of SA. At iteration $n\geq 1$ of SA, and  given a current location $\bx^{n-1}\in\setX$, a candidate value $\by^{n}$ is generated using the probability distribution $K(\bx^{n-1}, \dd\by)\in\mathcal{P}(\setX)$, where $K:\setX\rightarrow\mathcal{P}(\setX)$ is a Markov kernel. Then, a Metropolis step is performed to decide whether we move or not to this proposed location. In particular, if $\by^{n}$ increases the function value, i.e. $\varphi(\by^{n})\geq \varphi(\bx^{n-1})$, then we move with probability one to the location  $\by^{n}$, i.e. $\bx^{n}=\by^{n}$ almost surely. However,  some downwards moves are also accepted in order to escape quickly from local maxima. Given a current location $\bx^{n-1}$ and a candidate value $\by^{n}$, the probability to accept a downward move depends on $T_{n}$, the level of temperature at iteration $n$. The sequence $(T_n)_{n\geq 1}$ is strictly positive and converges to zero as $n\rightarrow\infty$, allowing for more exploration early in the algorithm.

There exist various convergence results for SA algorithms, the vast majority of them being based on the assumption that the state space is compact; see however \citet{Andrieu2001}   for  results on  non-compact  spaces. For compact  spaces, convergence results can be found, e.g., in \citet{Belisle1992, Locatelli1996,Locatelli2000b,Haario1991}; see also \citet{Lecchini2010} for results on the finite time behaviour of SA on compact spaces.

\subsection{Introduction to quasi-Monte Carlo}\label{sub:QMC}

As mentioned in the introduction, QMC point sets are sets of points in $\ui^d$ that are more evenly distributed  than IID uniform random numbers. There exist in the QMC literature many different measures of uniformity  of a point set $\bu^{1:N}$ in $\ui^d$, the most popular of them being the \emph{star discrepancy}, defined as
$$
D^{\star}(\bu^{1:N})=\sup_{\bm{b}\in (0,1]^d}\Big|\sum_{n=1}^N\ind\big(\bu^n\in [\bm{0},\bm{b})\big)-\prod_{i=1}^db_i\Big|,
$$ 
where $\ind(\cdot)$ denotes the indicator function. A QMC (or low discrepancy) point set may be formally defined as a set of points $\bu^{1:N}$ in $\ui^d$ such that $D^{\star}(\bu^{1:N})=\bigO(N^{-1}(\log N)^d)$. Note that for a set of IID  uniform random numbers, $D^{\star}(\bu^{1:N})=\bigO(N^{-1/2}\log\log N)$ almost surely \citep[see, e.g.,][p.167]{Niederreiter1992}. There exist various constructions of QMC point sets $\bu^{1:N}$; see  for instance the books of \citet{Niederreiter1992,dick2010digital} for more details on this topic.

\subsubsection{Quasi-Monte Carlo optimization}\label{sub:QMCopt}

If QMC point sets are  widely used to evaluate the integral of a function, there has been remarkably little use of QMC for global optimization, with efforts primarily limited to QMC versions of random search. We start by discussing this work in order to motivate why and how the good equidistribution properties of QMC point sets may be useful to improve stochastic optimization methods.

Consider the problem of computing $\varphi^*$ for a function $\varphi:\ui^d\rightarrow\mathbb{R}$. Then, a simple way to approximate $\varphi^*$ is to compute $m(\varphi,\bu^{1:N}):=\max_{1\leq n\leq N}\varphi(\bu^n)$,
where  $\bu^{1:N}$ is a set  of $N$ points in $\ui^d$. \citet{Niederreiter1983b} shows that, for Lipschitz functions $\varphi$,
$$
\varphi^*-m(\varphi,\bu^{1:N})\leq C(\varphi) d\big(\bu^{1:N}\big)
$$
for a constant $C(\varphi)<\infty$ which depends only on $\varphi$ and where
$$
d\big(\bu^{1:N}\big)=\max_{\bx\in\ui^d}\min_{1\leq n\leq N}\|\bu^n-\bx\|_{\infty}
$$
is the \emph{dispersion} of $\bu^{1:N}$ on $\ui^d$, which measures the maximum distance between a point in the state space and the points used in the search. Intuitively, the good partition of the points of a QMC point set  inside the unit hypercube should translate into a small dispersion. And indeed, there exist several constructions of QMC point sets which verify $d\big(\bu^{1:N}\big)\leq \bar{C} N^{-{}1/d}$ for a constant $\bar{C}<\infty$, where  the dependence in $N\geq 1$ is optimal \citep[][Theorem 6.8, p.154]{Niederreiter1992}.

Ad hoc improvements of quasi-random search, designed to improve upon the non-adaptive $N^{-1/d}$ convergence rate, have been proposed by \citet{Niederreiter1986}, \citet{Hickernell1997}, \citet{Lei2002},  and \citet{Jiao2006}, see also \citet{Fang2002} and references therein. Finally, note that QMC optimization has been applied successfully  in statistics for, e.g., maximum likelihood estimation and parameter estimation in nonlinear regression models \citep[see][and references therein]{Fang1996, Fang2002} as well as for  portfolio management \citep{Pistovvcak2004} and power system tuning \citep{Alabduljabbar2008}.

In optimization, we often run the algorithm until a  given tolerance criterion is fulfilled, and thus the number of function evaluations is not known in advance. Therefore, we will  focus in this paper on QMC point sets $\bu^{1:N}$ obtained by taking the first $N$ points of a sequence $(\bu^{n})_{n\geq 1}$. Most constructions of QMC sequences belong to the class of the so-called $(t,s)$-sequences \citep{Niederreiter1987} which are, as shown below, well adapted for optimization and therefore constitute the key ingredient of QMC-SA.

\subsubsection{Definition and main properties of $(t,s)$-sequences\label{sub:(t-s)}}

For integers $b\geq 2$ and $s\geq 1$, let 
$$
\mathcal{E}_s^b=\Big\{\prod_{j=1}^s\big[a_j b^{-d_j},(a_j+1)b^{-d_j}\big)\subseteq \ui^s,\, a_j,\,d_j\in\mathbb{N},\, a_j< b^{d_j},
\, j\in 1:s\Big\}
$$
be the  set of all $b$-ary boxes (or elementary intervals in base $b$) in $\ui^s$. Then, we introduce the notion of $(t,s)$-sequences through two definitions:

\begin{definition}
Let $m\geq 0$, $b\geq 2$, $0\leq t\leq m$ and $s\geq 1$  be integers. Then,  the set $\{\bu^n\}_{n=0}^{b^m-1}$ of $b^m$ points in $\ui^s$ is called a $(t,m,s)$-net in base $b$ in every $b$-ary box $E\in \mathcal{E}_s^b$ of volume $b^{t-m}$ contains exactly $b^t$ points of the point set $\{\bu^n\}_{n=0}^{b^m-1}$.

\end{definition}

\begin{definition}
Let $b\geq 2$, $t\geq 0$, $s\geq 1$  be integers. Then, the sequence $(\bu^n)_{n\geq 0}$ of points in $\ui^s$ is  called a  $(t,s)$-sequence in base $b$ if, for any integers $a\geq 0$ and $m\geq  t$, the set  $\{\bu^n\}_{n=ab^m}^{(a+1)b^m-1}$ is a $(t,m,s)$-net in base $b$.
\end{definition}


An interesting property of $(t,s)$-sequences for optimization is that,  for any $N\geq 1$, the dispersion of the  first $N$ points of  a $(t,s)$-sequence is bounded by $C_{b,t,s}N^{-1/s}$ for a constant  $C_{b,t,s}<\infty$ \citep[][Theorem 6.11, p.156]{Niederreiter1992}, where the dependence in $N$ is optimal as recalled in the previous subsection.  Over $(t,m,s)$-nets, the dispersion is bounded by $C'_{b,t,s}b^{-m/s}$ for a constant $C'_{b,t,s}<C_{b,t,s}$ \citep[][Theorem 6.10, p.156 ]{Niederreiter1992}. Note that the integer $b$  determines how often sets of points with good coverage of the hypercube arrive as we go through  the sequence, while  the parameter $t\geq 0$ measures the quality of these sets. In particular, for a given value of $s$ and $b$, the smaller $t$  is the better the $(t,m,s)$-net spreads over the unit hypercube. However, for a given value of $b\geq 2$ and $s$, $(t,s)$-sequences exist only for some values of $t$ and, notably, a $(0,s)$-sequence exists only if $b\geq s$ \citep[see, e.g.,][Corollary 4.36, p.141]{dick2010digital}.  We refer the reader to \citet[][Chapter 4]{Niederreiter1992} and \citet[][Chapter 4]{dick2010digital} for a more complete exposition of $(t,s)$-sequences.

\section{Quasi-Monte Carlo simulated annealing\label{sec:QMC-SA}}

The QMC-SA algorithm we study  in this work is presented in Algorithm \ref{alg:SAQMC}. Note that, when the input sequences  are IID uniform random numbers in $\ui^{d+1}$,  Algorithm \ref{alg:SAQMC} reduces to standard SA  where we have written the simulations of the candidate value $\by^n$ and the Metropolis steps as a function of $d+1$ random numbers using the inverse Rosenblatt transformation of the Markov kernel.

To obtain a QMC version of SA, Algorithm \ref{alg:SAQMC} simply amounts to replacing the pseudo-random numbers used in classical simulated annealing algorithms by specific QMC sequences, whose choice  is crucial for the performance of the optimization routine. The SA algorithm  proposed in this work is based on what we call $(t,s)_R$-sequences, that we introduce in the the next subsection.

\subsection{$(t,s)_R$-sequences: Definition and their use in QMC-SA\label{sub:(t-s)_R}}
$ $
\begin{definition}
Let $b\geq 2$, $t\geq 0$, $s\geq 1$ and $R\in\bar{\mathbb{N}}$ be integers. Then, we say that the  sequence $(\bu_R^n)_{n\geq 0}$ of points in $\ui^{s}$ is a $(t,s)_R$-sequence in base $b$ if, for all $n\geq 0$ (using the convention that empty sums are null),
$$
\bu_R^n=(u_{R,1}^n,\dots,u_{R,s}^n),\quad u_{R,i}^n=\sum_{k=1}^{R}a_{ki}^nb^{-k}+b^{-R}z_i^n,\quad i\in 1:s
$$
where the $z_i^n$'s are IID uniform random variables in $\ui$ and where the digits $a_{ki}^n$'s in $\{0,\dots,b-1\}$ are such that  $(\bu_{\infty}^n)_{n\geq 0}$ is a $(t,s)$-sequence in base $b$. 
\end{definition}

The parameter $R$ therefore regulates the degree of randomness  of the sequence $(\bu_R^n)_{n\geq 0}$, with the two extreme cases $R=0$ and $R=\infty$ corresponding, respectively, to   a sequence of IID uniform random variables  and  a completely deterministic sequence. For $R\geq t$, note that, for all $a\in\mathbb{N}$ and for all $m\in t:R$, the set $\{\bu_R^n\}_{n=ab^m}^{(a+1)b^m-1}$ is a $(t,m,s)$-net in base $b$. In addition, for any $R<\infty$ and $n\geq 0$,  $\bu_R^n$ is uniformly distributed in a $b$-ary box $E_n\in\mathcal{E}_s^b$  of volume $b^{-sR}$, 
whose position in $\ui^s$ depends on the deterministic part of $\bu_{R}^n$.

In practice, one will often use   $R=0$ or the  limiting case $R=\infty$. But, to rule out some odd behaviours that cannot be excluded  due to the deterministic nature of $(t,s)$-sequences, our general consistency result only applies for $R\in\mathbb{N}$. As already mentioned, note that an arbitrary small degree of randomness suffices to guarantee the validity of QMC-SA since our consistency results holds for \textit{all} $R\in\mathbb{N}$.

However, in practice, the randomization can be omitted. Indeed, SA algorithms are often run for a maximum of $N<\infty$ iterations and, if the result is not satisfactory, instead of increasing $N$,  it is usually  run again using a different starting value  and/or a different input sequence. For most constructions of $(t,s)$-sequences (e.g. Sobol', Faure and  Niederreiter-Xing sequences),  $a_{ik}^n=0$ for all $k\geq k_n$, where $k_n$ denotes the smallest integer such that $n<b^{k_n}$. Thus, for a fixed $N$, if one chooses $R\geq k_N$, the randomization is not necessary. 

Note also that $(t,s)_R$-sequences contain the practical versions of scrambled $(t,s)$-sequences \citep{Owen1995}. A scrambled $(t,s)$-sequence $(\tilde{\bu}^n)_{n\geq 0}$ is obtained from a $(t,s)$-sequence $(\bu^n)_{n\geq 0}$ by  randomly permuting (`scrambling') the digits $a_{ki}^n$'s of its components in a way that preserves the  equidistribution properties of $(\bu^n)_{n\geq 0}$. In particular, the  random sequence $(\tilde{\bu}^n)_{n\geq 0}$ is a $(t,s)$-sequence in base $b$ with probability one and, in addition, $\tilde{\bu}^n\sim\Unif(\ui^{s})$ for all $n\geq 0$.  The  sequence $(\tilde{\bu}^n)_{n\geq 0}$ has component $\tilde{u}_i^n=\sum_{k=1}^{\infty}\tilde{a}_{ki}^nb^{-k}$ such that $\tilde{a}_{ki}^n\neq 0$ for all $k,i,s$ but, in practice, the infinite sum is truncated  and instead we use the sequence $(\check{\bu}^n)_{n\geq 0}$, with component $\check{u}_i^n=\sum_{k=1}^{K_{\text{max}}}\tilde{a}_{ki}^nb^{-k}+b^{-K_{\text{max}}}z_i^n$ where, as above, the $z_i^n$'s are IID uniform random variables in $\ui$ \citep[see][]{Owen1995}.

\begin{algorithm}
\begin{algorithmic}[1]
\caption{Generic QMC-SA algorithm to compute $\sup_{\bx\in\setX}\varphi(\bx)$\label{alg:SAQMC}}
\Require A starting point $\bx^0\in\setX$, a Markov kernel $K$ acting from $\setX$ to itself, $T_1>0$,  $(\bu_R^n)_{n\geq 0}$, a $(t,d)_R$-sequence in base $b\geq 2$ and $(v^n)_{n\geq 0}$ a $(0,1)$-sequence in base $b$.

\For{$n=1\to N$}

\State Compute  $\by^n=F_{K}^{-1}(\bx^{n-1},\bu_R^n)$
\If{$v^n \leq  \exp\{(\varphi(\by^n)-\varphi(\bx^{n-1}))/T_{n}\}$}
\State $\bx^n=\by^n$
\Else 
\State $\bx^n=\bx^{n-1}$
\EndIf
\State Select $T_{n+1}\in(0, T_n]$
\EndFor
\end{algorithmic}
\end{algorithm}

\subsection{Practical implementation\label{sub:implemenation}}

Algorithm \ref{alg:SAQMC} simply amounts to replacing the pseudo-random numbers used in classical simulated annealing algorithms by points taken from  a $(0,1)$-sequence  and from a $(t,s)_R$-sequences $(\bu_R^n)_{n\geq 0}$. In practice, the generation of $(t,s$)-sequences is   an easy task  because  most statistical software contain routines to generate them (e.g. the package  \texttt{randtoolbox} in R or  the class \texttt{qrandset} in the statistical toolbox of Matlab). Generating  $(\bu_R^n)_{n\geq 0}$ for $R\in\mathbb{N}_+$ is more complicated because one should act at the digit level. However, as a  rough proxy, one can generate $(\bu_R^n)_{n\geq 0}$ as follows: generate a  $(t,d)$-sequence $(\bu^n)_{n\geq 0}$ and set $\bv_R^n=\bu^n+b^{-R-1}\bm{z}^n$, where the $\bm{z}^n$'s are IID uniform random variables in $\ui^d$.

Concerning the computational cost,  generating the first $N$ points  of most constructions of $(t,s)$-sequences requires $\bigO(N\log N)$ operations \citep{Hong2003}. This is slower than random sampling but the extra cost is particularly small when $b=2$ since, in that case, bits operations can be used to generate the points of the sequence.

Steps 2-7 of Algorithm \ref{alg:SAQMC} sample $\bx^n$ from a distribution $\tilde{K}(\bx^{n-1},\dx)\in\mathcal{P}(\setX)$ using the point $(\bu_R^n,v^n)\in\ui^{d+1}$. Thus, although not required for our consistency results, it seems a good idea to choose the $(0,1)$-sequence $(v^n)_{n\geq 0}$ and the limiting $(t,d)$-sequence  $(\bu_{\infty}^n)_{n\geq 0}$ such that the sequence $(\bv^n)_{n\geq 0}$ in $\ui^{d+1}$, with component $\bv^n=(\bu_{\infty}^n,v^n)$, is a $(t,d+1)$-sequence. Note that this is a weak requirement because  most standard constructions of $(t,d+1)$-sequences  $((\bu_{\infty}^n,v^n))_{n\geq 0}$, such as, e.g.,  the  Sobol' or the Faure sequences \citep[see][Chapter 8, a definitions]{dick2010digital},  are such that $(\bu_{\infty}^n)_{n\geq 0}$ and $(v^n)_{n\geq 0}$ have the right properties.

Finally, compared to Monte Carlo SA, QMC-SA has the drawback of requiring Markov kernels $K$ that  can be sampled efficiently using the inverse Rosenblatt transformation approach. However, this is not an important practical limitation because SA algorithms are often based on (truncated)  Gaussian or Cauchy kernels for which it is trivial to compute the inverse Rosenblatt transformation. Again, we refer the reader e.g. to \citet[][Section 3.1]{SQMC} for a detailed presentation of this sampling strategy.

\section{Convergence results\label{sec:convergence}}

The convergence study of Algorithm \ref{alg:SAQMC} is  conducted under the assumption  that $\setX$ is a non-empty closed hyperrectangle of $\mathbb{R}^d$. Without loss of generality,  we assume that $\setX=[0,1]^d$ in what follows.

Our results require some (weak) regularity assumptions on the Markov kernel in order to preserve the good equidistribution properties of the QMC sequences. More precisely, we consider the following two assumptions on $K:[0,1]^d\rightarrow\mathcal{P}([0,1]^d)$ :

\begin{enumerate}[label=(\subscript{A}{\arabic*})]
\item\label{H:lem:K1}  For a fixed $\bx\in\setX$, the $i$-th component of  $F_{K}(\bx,\by)$ is strictly increasing in $y_i\in[0,1]$, $i\in 1:d$;
\item\label{H:lem:K2} The Markov kernel $K(\bx,\dd \by)$ admits a continuous density $K(\by|\bx)$ (with respect to the Lebesgue measure on $[0,1]^d$) such that, for a constant $\underline{K}>0$,   $K(\by|\bx) \geq \underline{K}$ for all $(\bx,\by)\in \setX^2$.
\end{enumerate}
Assumption \ref{H:lem:K1} ensures that, for any $\bx\in\setX$, the inverse Rosenblatt transformation of $F_K(\bx,\dd\by)\in\mathcal{P}(\setX)$ is a well defined function while Assumption \ref{H:lem:K2} is  used, e.g., in \citet{Belisle1992}.


In the next subsection we  state some preliminary results on the dispersion of point sets obtained by transforming a $(t,s)$-sequence through inverse Rosenblatt transformations. These results provide key insight into how QMC may improve the performance of SA; readers mostly interested in the main results can however skip  this  subsection and go directly to Sections \ref{sub:consisistency} and \ref{sub:R0}.

\subsection{Preliminary results\label{sub:prel}}

The following lemma provides conditions on the Markov kernel $K:\setX\rightarrow\mathcal{P}(\setX)$ so that the inverse Rosenblatt transformation $F_K^{-1}(\bx,\cdot)$ converts a low dispersion point set in $\ui^d$ into a low dispersion point set in $\setX$.

\begin{lemma}\label{lem:dense}
Let $\setX=[0,1]^d$ and $(\bu^n)_{n\geq 0}$ be a $(t,d)$-sequence in base $b\geq 2$. Let $K:\setX\rightarrow\mathcal{P}(\setX)$ be a Markov kernel such that Assumptions \ref{H:lem:K1}-\ref{H:lem:K2} hold. Let $(\tilde{\bx},\bx')\in\setX^2$ and $\by^n=F_K^{-1}(\tilde{\bx},\bu^n)$, $n\geq 0$. Then, there exists a $\bar{\delta}_K>0$  such that, for any $\delta\in (0,\bar{\delta}_K]$, there is a $k_{\delta}\in\mathbb{N}$, $k_{\delta}\geq t+d$,  such that, for any $a\in\mathbb{N}$, the point set $\{\by^n\}_{n=ab^{k_{\delta}}}^{(a+1)b^{k_{\delta}}-1}$ contains at least $b^t$ points in $B_{\delta}(\bx')$. In addition, $\bar{\delta}_K$ and $k_{\delta}$ do not depend on the point $(\tilde{\bx},\bx')\in\setX^2$. Moreover, the   result  remains true if $\by^n=F_K^{-1}(\bx^n,\bu^n)$ with $\bx^n\in B_{v_K(\delta)}(\tilde{\bx})$ for all $n\geq 0$ and where $v_K:(0,\bar{\delta}_K]\rightarrow\mathbb{R^+}$ is independent of $(\tilde{\bx},\bx')\in\setX^2$, continuous and strictly increasing, and such that $v_K(\delta)\cvz$ as $\delta\cvz$.
\end{lemma}
\begin{proof}
See Appendix \ref{p-lem:dense} for a proof.
\end{proof}

As a corollary, note the following. Let $P_{b,k}=\{\bu^n\}_{n=0}^{b^k-1}$, $k\geq t$, be a $(t,k,d)$-net in base $b\geq 2$ and define
\begin{align}\label{eq:dense}
d_{\setX}(P_{b,k},\tilde{\bx},K)=\sup_{\bx\in\setX}\min_{n\in 0: b^k-1}\|F_K^{-1}(\tilde{\bx},\bu^n)-\bx\|_{\infty},\quad \tilde{\bx}\in\setX
\end{align}
as the dispersion of the point set $\{F_K^{-1}(\tilde{\bx},\bu^n)\}_{n=0}^{b^k-1}$ in $\setX$. Then, under the conditions of the previous lemma,  $d_{\setX}(P_{b,k},\tilde{\bx},K)\cvz$ as $k\rightarrow\infty$, uniformly on $\tilde{\bx}\in\setX$.

Lemma \ref{lem:dense} provides an iterative random search interpretation of Algorithm \ref{alg:SAQMC}. Indeed, at location $\bx^n=\tilde{\bx}$, this latter go through a low dispersion sequence in the state space $\setX$   until a candidate value  is accepted, and then the process starts again at this new location. The last part of the lemma shows that it is not a good idea to re-start the sequence each time we move to a new location because  the point set $\{F_K^{-1}(\bx^{n},\bu^{n})\}_{n=0}^{b^k-1}$ may have good dispersion properties in $\setX$ even when  $\bx^{n'}\neq\bx^n$ for some $n,n'\in 0:(b^k-1)$.

It is also worth noting that Lemma \ref{lem:dense} gives an upper bound for the jumping time of the deterministic version of Algorithm \ref{alg:SAQMC}; that is, for the number $m_n$ of candidate values  required to move from location $\bx^n=\tilde{\bx}$ to a new location $\bx^{n+m_n}\neq\tilde{\bx}$. Indeed, let $\bx^*\in\setX$ be a global maximizer of $\varphi$, $\delta_n=\|\bx^n-\bx^*\|_{\infty}$ and assume that there exists a $\delta\in(0,\delta_n)$  such that $\varphi(\bx)>\varphi(\bx^n)$ for all $\bx\in B_{\delta}(\bx^*)$. Then, by the properties of $(t,s)$-sequences (see Section \ref{sub:(t-s)}), the point set $\{\bu_{\infty}^m\}_{m=n+1}^{n+2k_{\delta}}$, with $k_\delta$ as in Lemma \ref{lem:dense}, contains at least one $(t,k_{\delta},s)$-net and thus, by this latter, there exists at least one $m\in 1:2k_{\delta}$ such that $\tilde{\by}^{n+m}=F_{K}^{-1}(\tilde{\bx},\bu_{\infty}^{n+m})$ belongs to $B_{\delta}(\bx^*)$. Since $\varphi(\tilde{\by}^m)>\varphi(\tilde{\bx})$, we indeed have $m_n\leq 2k_{\delta}<2k_{\delta_n}$.

It is clear that the speed at which $d_{\setX}(P_{b,k},\tilde{\bx},K)$ goes to zero as $k$ increases depends on the smoothness of the transformation $F^{-1}_{K}(\bx,\cdot)$, $\bx\in\setX$. Thus, it is sensible to choose $K$ such that $d_{\setX}(P_{b,k},\tilde{\bx},K)$ goes to zero at the optimal  non-adaptive convergence rate (see Section \ref{sub:QMCopt}). The next lemma provides sufficient conditions for the kernel $K$ to fulfil this requirement.

\begin{lemma}\label{lem:dense2}
Consider the set-up of Lemma \ref{lem:dense} and, in addition, assume that, viewed as a function of $(\bx,\by)\in\setX^2$, $F_K(\bx,\by)$ is Lipschitz with Lipschitz constant $C_K<\infty$ for the  supremum norm. Let 
$$
\tilde{C}_{K}=0.5\tilde{K}\big(1 \wedge (0.25\tilde{K}/C_K)^{d}\big),\quad \tilde{K}=\min_{i\in 1:d}\{\inf_{(\bx,\by)\in \setX^2}K_i(y_i|y_{1:i-1},\bx)\}.
$$ 
Then, the   result of Lemma \ref{lem:dense} holds for   $
k_{\delta}=\big\lceil t+d-d\log(\delta \tilde{C}_{K}/3)/\log b\big\rceil$ and for $\bar{\delta}_K= (3/\tilde{C}_K)\wedge 0.5$.
\end{lemma}
\begin{proof}
See Appendix \ref{p-lem:dense2} for a proof.
\end{proof}

Let $\delta_k\in (0,\bar{\delta}_K]$ be the size of the smallest ball around $\bx'\in\setX$ that can be reached by the point set $\{F_K^{-1}(\tilde{\bx},\bu^n)\}_{n=0}^{b^k-1}$, with $k\geq k_{\bar{\delta}_{K}}$, $\tilde{\bx}\in\setX$ and where, as above, $P_{b,k}=\{\bu^n\}_{n=0}^{b^k-1}$ is a $(t,k,d)$-net in base $b\geq 2$. Then, under the assumptions of Lemma \ref{lem:dense2}, $\delta_k=Cb^{-k/d}$ for a constant $C>0$ independent of $\bx'$ and thus $
\sup_{\tilde{\bx}\in\setX}d_{\setX}(P_{b,k},\tilde{\bx},K)=\bigO(b^{-k/d})$
as required.

\subsection{General consistency results for  QMC-SA\label{sub:consisistency}}
 
To obtain a convergence result for Algorithm \ref{alg:SAQMC} we need some regularity assumptions concerning the objective function $\varphi$. To this end, and borrowing an idea of \citet{He2015} (who study QMC integration of functions restricted to a non-rectangular set), we impose a condition on the Minkovski content of the level sets 
$$
\setX_{\varphi(\bx)}:=\{\bx'\in\setX:\varphi(\bx')=\varphi(\bx)\},\quad \bx\in\setX.
$$

\begin{definition}
The measurable set $A\subseteq\setX$  has a $(d-1)$-dimensional Minkovski content if
$M(A):=\lim_{\epsilon \downarrow 0}\epsilon^{-1} \lambda_d\big((A)_{\epsilon}\big)$ exists and is finite, 
where, for $\epsilon>0$,  we use the shorthand $(A)_{\epsilon}=\{\bx\in\setX:\exists \bx'\in A, \|\bx-\bx'\|_{\infty}\leq\epsilon\}$.
\end{definition}

The following lemma is the key result to establish the consistency of QMC-SA.

\begin{lemma}\label{lem:convPhi}
Consider Algorithm \ref{alg:SAQMC} where $\setX=[0,1]^d$ and assume the following conditions are verified:
\begin{enumerate}
\item The Markov kernel $K:\setX\rightarrow\mathcal{P}(\setX)$ is such that Assumptions \ref{H:lem:K1}-\ref{H:lem:K2} hold;
\item $\varphi$ is continuous on $\setX$ and such that $
\sup_{x\in\setX:\,\varphi(x)<\varphi^*}M(\setX_{\varphi(\bx)})<\infty$.
\end{enumerate} 
Let $R\in\mathbb{N}$. Then, if the sequence $(\varphi(\bx^n))_{n\geq 1}$ is almost surely convergent,   with probability one $\varphi(\bx^n)\rightarrow \varphi^*$ as $n\rightarrow\infty$.
\end{lemma}

\begin{proof}
See Appendix \ref{p-lem:convPhi} for a proof.
\end{proof}

Note that the assumption on the Minkovski content of the level sets notably implies that  $\varphi$ has a finite number of modes.

Thus, if in Algorithm \ref{alg:SAQMC} only upward moves are accepted, then the resulting ascendant algorithm is convergent. To establish the consistency of QMC-SA, we therefore need to ensure that not too many ``large'' downward moves  are accepted. This is done by controlling the rate at which the sequence of temperature $(T_n)_{n\geq 1}$ converges toward zero, as shown in the next  results. We recall that     $k_n$ denotes the smallest integer such that $n<b^{k_n}$.

\begin{lemma}\label{lem:Tn}
Consider Algorithm \ref{alg:SAQMC} with $\setX\subseteq\mathbb{R}^d$ and where  $(v^n)_{n\geq 0}$ is such that  $v^0=0$. Then, at iteration $n\geq 1$,  $\by^n$ is rejected if $\varphi(\by^n)<\varphi(\bx^{n-1})-T_nk_n\log b$.
\end{lemma}
\begin{proof}
To prove the lemma, note that for any $k\in\mathbb{N}$ the point set $\{v^n\}_{n=0}^{b^k-1}$ is a $(0,k,1)$-net in base $b$ and therefore contains exactly one point in each elementary interval of length $b^{-k}$. Hence, because $v^0=0$,  the point set $\{v^n\}_{n=1}^{b^{k}-1}$ contains no point in the interval $[0, b^{-k})$ and thus, for all $n\geq 1$, $\by^n$ is rejected if 
$$
\exp\big(\big\{\varphi(\by^n)-\varphi(\bx^{n-1})\big\}/T_n\big)<b^{-k_n}.
$$
\end{proof}

\begin{lemma}\label{lem:ln}
Consider Algorithm \ref{alg:SAQMC} with $\setX\subseteq\mathbb{R}^d$ and assume that there exists a sequence $(l_n)_{n\geq 1}$ of positive numbers which verifies $\sum_{n=1}^{\infty}l_n^{-1}<\infty$ and such that $\bx^{n}=\bx^{n-1}$ when $\varphi(\by^n)< \varphi(\bx^{n-1})-l_n^{-1}$. Then, if  $\varphi$ is bounded, the sequence $(\varphi(\bx^n))_{n\geq 0}$ is almost surely convergent.
\end{lemma} 

\begin{proof}
Let $\bar{\varphi}=\limsup_{n\rightarrow\infty}\varphi(\bx^n)$ and $\underline{\varphi}=\liminf_{n\rightarrow\infty}\varphi(\bx^n)$. Then, because  $\varphi$ is bounded, both $\bar{\varphi}$ and $\underline{\varphi}$ are finite. We now show that, under the conditions of the lemma, we  have $\underline{\varphi}=\bar{\varphi}$ with probability one. In what follows, although not explicitly mentioned to save space, all the computations should be understood as  holding with probability one.

Let  $(\bx^{s_n})_{s_n\geq 0}$ and  $(\bx^{u_n})_{u_n\geq 0}$ be two subsequences such that $\varphi(\bx^{s_n})\rightarrow\bar{\varphi}$ and $\varphi(\bx^{u_n})\rightarrow\underline{\varphi}$. Let $\epsilon>0$ and $N_{\epsilon}\geq 1$ be such that $|\varphi(\bx^{s_n})-\bar{\varphi}|\leq 0.5\epsilon$ and $|\varphi(\bx^{u_n})-\underline{\varphi}|\leq 0.5\epsilon$ for all $n\geq N_{\epsilon}$. Then,  for
all $n >s_{N_{\epsilon}}$, we have
$$
\varphi(\bx^n)\geq  \varphi(\bx^{s_{N_{\epsilon}}})-\sum_{i=s_{N_{\epsilon}}}^{\infty}l_i^{-1}\geq \bar{\varphi}-0.5\epsilon-\sum_{i=s_{N_{\epsilon}}}^{\infty}l_i^{-1}
$$
and, in particular, $\bar{\varphi}-0.5\epsilon-\sum_{n=s_{N_{\epsilon}}}^{\infty}l_n^{-1}\leq  \varphi(\bx^{u_n})\leq  \underline{\varphi}+0.5\epsilon$,  $\forall u_n> \max(s_{N_{\epsilon}}, u_{N_{\epsilon}})$, implying that $\bar{\varphi}-\underline{\varphi}\leq\epsilon+ \sum_{i=s_{N_{\epsilon}}}^{\infty}l_i^{-1}$. In addition,  the series $\sum_{n=1}^{\infty}l_n^{-1}$ is convergent and thus $\sum_{i=s_{N_{\epsilon}}}^{\infty}l_i^{-1}\cvz$ as $s_{N_{\epsilon}}$ increases. Therefore, there exists a $\epsilon>0$ and a $N_{\epsilon}>0$ for which $\bar{\varphi}-\underline{\varphi}\leq \epsilon+ \sum_{i=s_{N_{\epsilon}}+1}^{\infty}l_i^{-1}\leq 0.5(\bar{\varphi}-\underline{\varphi})$, showing that we  indeed have $\underline{\varphi}=\bar{\varphi}$.
\end{proof}

Using Lemmas \ref{lem:convPhi}-\ref{lem:ln} we deduce the following general consistency result for Algorithm \ref{alg:SAQMC}:

\begin{theorem}\label{thm:conv}
Consider Algorithm \ref{alg:SAQMC} where $\setX=[0,1]^d$. Assume that the assumptions of Lemma \ref{lem:convPhi} hold and that, in addition, $(v^n)_{n\geq 0}$ is such that  $v^0=0$. Then, if $(T_n)_{n\geq 1}$ is  such that $\sum_{n=1}^{\infty}T_{n}\log(n)<\infty$, we have, for any $R\in\mathbb{N}$ and as $n\rightarrow\infty$, $\varphi(\bx^n)\rightarrow\varphi^*$ almost surely.
\end{theorem}

\begin{proof}
Let $l_n=(T_{n}k_n\log b)^{-1}$ and note that, under the assumptions of the theorem, $\sum_{n=1}^{\infty}l_n^{-1}<\infty$. Therefore,  by Lemma \ref{lem:Tn},  the sequence $(l_n)_{n\geq 0}$ verifies the assumptions of Lemma \ref{lem:ln}. Hence, because the continuous function $\varphi$ is bounded on the compact space $\setX$, the sequence $(\varphi(\bx^n))_{n\geq 0}$  is almost surely convergent by Lemma \ref{lem:ln} and thus converges almost surely toward the global maximum $\varphi^*$ by Lemma \ref{lem:convPhi}.
\end{proof}

As already mentioned, this result does not apply for $R=\infty$ because at this degree of generality we cannot  rule out the possibility of some odd behaviours when  completely deterministic sequences are used as input of QMC-SA. However, when $d=1$, things become easier and we can show that the result of Theorem \ref{thm:conv} holds for the sequence $(\bu_{\infty}^n)_{n\geq 0}$

\begin{theorem}\label{thm:conv_Univ}
Consider Algorithm \ref{alg:SAQMC} where $\setX=[0,1]$ and $R=\infty$. Assume that $\varphi$, $K$ and $(T_n)_{n\geq 1}$ verify the conditions of Theorem \ref{thm:conv}. Then, $\varphi(x^n)\rightarrow\varphi^*$ as $n\rightarrow\infty$.
\end{theorem}

\begin{proof}
See Appendix \ref{p-thm:conv_Univ} for a proof.
\end{proof}

Remark that, when $d=1$, the condition on the Minkovski content of the level sets of $\varphi$ given in Lemma \ref{lem:convPhi} amount to assuming that, for all $x\in\setX$ such that $\varphi(x)< \varphi^*$, there exists a $\delta_{x}>0$ such that $\varphi(y)\neq\varphi(x)$ for all $y\in B_{\delta_{x}}(x)$, $y\neq x$.

A key feature of Theorems \ref{thm:conv} and  \ref{thm:conv_Univ} is that the condition on the cooling schedules $(T_n)_{n\geq 1}$ depends neither on $\varphi$ nor on the choice of the Markov kernel $K$, while those necessary to ensure the convergence of standard Monte Carlo SA algorithms usually do. This is for instance the case for  the almost sure convergence result of \citet[][Theorem 3]{Belisle1992}; see the next subsection.  An open question for future research is to establish if this property of QMC-SA also holds on non-compact spaces, where convergence of SA is  ensured for a sequence $(T_n)_{n\geq 1}$ where $T_n=T_0/\log (n+C)$ and where both $T_0>0$ and $C>0$ are model dependent \citep[see][Theorem 1]{Andrieu2001}.  The simulation results presented below seem to support this view (see Section \ref{sub:spatial}).

Finally, note that the cooling schedule may be adaptive, i.e. $T_{n+1}$ may depend on $\bx^{0:n}$, and that the convergence rate for $(T_n)_{n\geq 1}$ implied by Theorem \ref{thm:conv_Univ} is coherent with \citet[][Theorem 2]{Belisle1992} which shows that almost sure convergence cannot hold if 
$\sum_{n=1}^{\infty}\exp(-T_n^{-1})=\infty$.

\subsection{The special case $R=0$ and a new convergence result for SA\label{sub:R0}}

It is worth noting that Theorem \ref{thm:conv} applies for $R=0$; that is, when IID uniform random numbers are used to generate the candidate values at Step 2 of Algorithm \ref{alg:SAQMC}. Remark that, in this case, the use of the inverse Rosenblatt transformation to sample from the Markov kernel is  not needed. In addition, it is easy to see from the proof of this result that the assumption on $\varphi$ can be weakened considerably. In particular, the continuity of $\varphi$ in the neighbourhood of one of its global maximizer is enough to ensure that, almost surely, $\varphi(\bx^n)\rightarrow\varphi^*$ (see the next result).

Since  Theorem \ref{thm:conv} also applies when  $R=0$, the key to remove the dependence of the cooling schedule to the problem at hand  therefore comes from the sequence $(v^n)_{n\geq 0}$, used in the Metropolis step, which  discards candidate values $\by^n$ which are such that $\varphi(\by^n)-\varphi(\bx^n)$ is too small (see  Lemma \ref{lem:Tn} above). This observation allows us to propose a new consistency result for Monte Carlo SA algorithms on compact spaces; that is, when Step 2 of Algorithm \ref{alg:SAQMC} is replaced by:

\begin{quote}
\centering
  2': Generate $\by^n\sim K(\bx^{n-1},\dd\by)$
\end{quote}
and when  $(v^n)_{n\geq 0}$ is replaced by a sequence of IID uniform random numbers in $\ui$.

\begin{theorem}\label{thm:convSA}
Consider Algorithm \ref{alg:SAQMC} with $R=0$, $\setX\subset\mathbb{R}^d$  a bounded measurable set, Step 2  replaced by Step 2' and   the sequence $(v^n)_{n\geq 0}$   replaced by $(\tilde{v}^n)_{n\geq 0}$, a sequence of IID uniform random numbers in $\ui$. Assume that the Markov kernel $K:\setX\rightarrow\mathcal{P}(\setX)$ verifies Assumption \ref{H:lem:K2} and that there exists a $\bx^*\in\setX$ such that $\varphi(\bx^*)=\varphi^*$ and such that $\varphi$ is continuous on  $B_{\delta}(\bx^*)$ for a $\delta>0$. Then, if $(T_n)_{n\geq 1}$ satisfies $\sum_{n=1}^{\infty}T_{n}\log(n)<\infty$, we have,  as $n\rightarrow\infty$,  $\varphi(\bx^n)\rightarrow\varphi^*$ almost surely.
\end{theorem}

\begin{proof}
Let $\alpha >0$ so that  $\sum_{n=1}^{\infty}\Prob(\tilde{v}^n<n^{-(1+\alpha)})=\sum_{n=1}^{\infty}n^{-(1+\alpha)}<\infty$. Thus, by Borel-Cantelli Lemma, with probability one  $\tilde{v}^n\geq n^{-(1+\alpha)}$ for all $n$ large enough. Therefore, for $n$ large enough and using similar arguments as in the proof of Lemma \ref{lem:Tn}, $\by^n$ is rejected with probability one  if $\varphi(\by^n)<\varphi(\bx^{n-1})-T_n(1+\alpha)(\log n)$ and hence, by Lemma \ref{lem:ln}, with probability one there exists a $\bar{\varphi}\in\mathbb{R}$ such that $\varphi(\bx^n)\rightarrow\bar{\varphi}$. To show that we almost surely have $\bar{\varphi}=\varphi^*$, let $\bx^*$ be as in the statement of the theorem. Then, using the fact that the Markov kernel verifies Assumption \ref{H:lem:K2}, it is easy to see that, with probability one, for any $\epsilon\in\mathbb{Q}_+$  the set $B_{\epsilon}(\bx^*)$ is visited infinitely many time by the sequence $(\by^n)_{n\geq 1}$. Then, the result follows from the continuity of $\varphi$ around $\bx^*$.
\end{proof}

The conditions on $\setX$ and on $\varphi$ are the same as in the almost sure convergence result of \citet[][Theorem 3]{Belisle1992}  but the condition on the Markov kernel is weaker. Finally, this latter result requires that $T_n\leq 1/(n\,e_n)$ for a model-dependent and strictly increasing sequence $(e_n)_{n\geq 1}$, while Theorem \ref{thm:convSA} establishes the almost sure convergence of SA for a universal sequence of temperature.

\section{A general class of QMC-SA type algorithms}\label{sec:extension}

We saw in Section \ref{sub:consisistency} that, if only upward moves are accepted in Algorithm \ref{alg:SAQMC}, then the resulting ascendant algorithm is  convergent. Thus, if one wishes to incorporate  downward moves in the course of the algorithm, we need to control their size and their frequency. In SA algorithms, downward moves are introduced  through the Metropolis step; suitable assumptions on the sequence of temperatures $(T_n)_{n\geq 1}$  guarantee the convergence of the algorithm. 

Interestingly, our convergence results extend to Algorithm \ref{alg:TAQMC} where a more general device is used to introduce downward moves.

\begin{corollary}\label{cor:TA}
Consider Algorithm \ref{alg:TAQMC} where $\setX=[0,1]^d$ and assume that $K$ and $\varphi$ verify the  assumptions of Theorem \ref{thm:conv}. Then, if the sequence $(l_n)_{n\geq 1}$ is such that $\sum_{n=1}^{\infty}l_n^{-1}<\infty$, we have,  for any $R\in\mathbb{N}$ and as $n\rightarrow\infty$, $\varphi(\bx^n)\rightarrow\varphi^*$ almost surely. In addition, if $d=1$, $\varphi(x^n)\rightarrow\varphi^*$ when $R=\infty$.
\end{corollary}
\begin{proof}
The result for $d>1$ is a direct consequence of Lemmas \ref{lem:convPhi} and \ref{lem:ln}, while the result for $d=1$ can be deduced from Lemma \ref{lem:ln} and from the proof of Theorem \ref{thm:conv_Univ}.
\end{proof}

\begin{algorithm}
\begin{algorithmic}[1]
\caption{Generic QMC-SA type algorithm to compute $\sup_{\bx\in\setX}\varphi(\bx)$\label{alg:TAQMC}}
\Require A starting point $\bx^0\in\setX$, a Markov kernel $K$ acting from $\setX$ to itself, $l_1>0$ and $(\bu_R^n)_{n\geq 0}$, a $(t,d)_R$-sequence in base $b\geq 2$.

\For{$n=1\to N$}

\State Compute  $\by^n=F_{K}^{-1}(\bx^{n-1},\bu_R^n)$
\If{$\varphi(\by^n)\geq \varphi(\bx^{n-1})$}
\State $\bx^n=\by^n$
\ElsIf{$\varphi(\by^n)< \varphi(\bx^{n-1})-l_n^{-1}$}
\State $\bx^n=\bx^{n-1}$
\Else
\State Set $\bx^n=\bx^{n-1}$ or $\bx^n=\by^{n}$ according to some rule
\EndIf
\State Select $l_{n+1}\in[l_n,\infty)$
\EndFor
\end{algorithmic}
\end{algorithm}

The fact that, under the assumptions on $(v^n)_{n\geq 0}$ and on $(T_n)_{n\geq 1}$ of Theorems \ref{thm:conv} and \ref{thm:conv_Univ}, Algorithm \ref{alg:SAQMC} indeed reduces to a particular case of Algorithm \ref{alg:TAQMC} which verifies the assumptions of Corollary \ref{cor:TA} is a direct consequence of Lemma \ref{lem:Tn}.

The modification of the exponential acceptance probability function to accelerate the convergence of the algorithm is a classical problem in the SA literature, see e.g. \cite{Rubenthaler2009} and references therein. Regarding this point, an important aspect of  the connection between  QMC-SA and Algorithm \ref{alg:TAQMC} is that, in the Metropolis step of Algorithm \ref{alg:SAQMC}, we can replace the exponential function by any strictly increasing function $f:\mathbb{R}\rightarrow[0,1]$ which satisfies $f(0)=1$ and $\lim_{t\rightarrow -\infty}f(t)=0$, and the convergence results of Section \ref{sec:convergence} remain valid under the simple condition that $(T_n)_{n\geq 1}$ verifies
$\sum_{n=1}^{\infty}\big(T_nf^{-1}(b^{-k_n})\big)<\infty$.

Another interesting case of Algorithm \ref{alg:TAQMC} is  the (QMC version of) Threshold Accepting (TA) algorithms introduced by  \citet{dueck1990,Moscato1990}, where in Step 8 we always set $\bx^n=\by^n$. Like SA, TA  has been introduced for optimization problems on finite state spaces but has been successfully applied for continuous problems in various fields, for instance see \citet{Winker2007} and references therein for applications of TA in statistics and in economics.

If convergence properties of SA are now well established, to the best of our knowledge the only convergence result for TA is the one of \citet{althofer1991}, obtained for optimization problems on finite spaces. However, their proof is not constructive in the sense that they only prove the existence of a sequence of thresholds $(l_n^{-1})_{n\geq 1}$  that provides convergence within a ball of size $\epsilon$ around the global solution. To this regards, Corollary \ref{cor:TA} therefore constitutes the first convergence result for this class of algorithms.

\section{Numerical Study\label{sec:num}}

The objective of this section is to compare the performance of classical SA  (i.e. Algorithm \ref{alg:SAQMC} based on IID random numbers) with QMC-SA  to find the global optimum of functions defined on  continuous spaces.

To compare SA and QMC-SA for a large number of different configurations, we first  consider a bivariate toy example for which we perform an extensive simulation study (Section \ref{sub:toy}). Then, SA and QMC-SA are compared on a difficult optimization problem arising in spatial statistics and which involves the minimization of a non-differentiable function defined on an unbounded space of dimension $d=102$ (Section \ref{sub:spatial}).

In all the simulations below the comparison between  SA and QMC-SA is based on 1\,000 starting values sampled independently in the state space. The Monte Carlo  algorithm is run only once while QMC-SA  is implemented using a Sobol' sequence as input (implying that $b=2$ and $R=\infty$).

\subsection{Example 1: Pedagogical example\label{sub:toy}}

Following \citet[][Example 5.9, p.169]{RobCas} we consider the  problem of minimizing the function $\tilde{\varphi}_1:\setX_1:=[-1,1]^2\rightarrow\mathbb{R}^+$  defined by 
\begin{equation}\label{eq:toy}
\begin{split}
\tilde{\varphi}_1(x_1,x_2)&=\big(x_1\sin(20x_2)+x_2\sin(20x_1)\big)^2\cosh\big(\sin(10x_1)x_1)\\
&+\big(x_1\cos(10x_2)-x_2\sin(10x_1)\big)^2\cosh\big(\sin(20x_2)x_2).
\end{split}
\end{equation}
This function has a unique global minimum at $(x_1,x_2)=(0,0)$ and several local minima on $\setX_1$; see \citet[][p.161]{RobCas} for a grid representation of  $\tilde{\varphi}_1$. The comparison between SA and QMC-SA for this  optimization problem is based on the number of iterations that is needed to reach a ball of size $10^{-5}$ around the global minimum of $\tilde{\varphi}_1$.  To avoid infinite running time, the maximum number of iterations we allow is $N=2^{17}$.

The SA and QMC-SA algorithms are implemented for Markov kernels
$$
K^{(j)}(\bx,\dd\by)=f^{(j)}_{[-1,1]}(y_1,x_1,\sigma)f^{(j)}_{[-1,1]}(y_2,x_2,\sigma)\dd \by,\quad j\in 1=1,2,
$$
where, for $j=1$ (resp. $j=2$),  $f^{(j)}_{I}(\cdot ,\mu,\tilde{\sigma})$
denotes  the density of the  Cauchy (resp. Gaussian) distribution with location parameter $\mu$ and scale parameter $\tilde{\sigma}>0$, truncated on $I\subseteq\mathbb{R}$. Simulations are done for $\sigma\in\{0.01,0.1,1,10\}$. 

We consider  three different sequences of temperatures $(T^{(m)}_n)_{n\geq 1}$, $m\in 1:3$, defined by
\begin{align}\label{num:Temp}
T^{(1)}_n=T^{(1)}_0(n^{1+\epsilon}\log n)^{-1},\quad T^{(2)}_n=T^{(2)}_0/n,\quad T^{(3)}_n=T_0^{(3)}(\log n)^{-1}
\end{align}
and where we choose $\epsilon=0.001$ and
\begin{align}\label{num:Temp2}
T^{(1)}_0\in\{20,200,2\,000\},\quad T^{(2)}_0\in\{2,20,200\},\quad T^{(3)}_0\in\{0.02,0.2,2\}.
\end{align}
Note that $(T^{(1)}_n)_{n\geq 1}$ is such that results of Theorems \ref{thm:conv}-\ref{thm:convSA} hold  while $(T^{(2)}_n)_{n\geq 1}$ (resp. $(T^{(3)}_n)_{n\geq 1}$ ) is the standard choice for SA based on  Cauchy (resp. Gaussian) random walks  \citep[see, e.g.,][]{Ingber1989}. However, on compact state spaces, SA based on these two kernels is such that the sequence $(\tilde{\varphi}_1(\bx^n))_{n\geq 1}$ converges in probability to the global minimum of $\tilde{\varphi}_1$ for any sequence of temperatures $(T_n)_{n\geq 1}$ such that $T_n\rightarrow 0$ as $n\rightarrow\infty$ \cite[see][Theorem 1]{Belisle1992}.

Simulations are performed for different combinations of kernels $K^{(j)}$ and temperatures $(T_n^{(m)})_{n\geq 1}$. For each of these combinations, simulations are done for all values of $T_0^{(m)}$ given in \eqref{num:Temp2} and for all $\sigma\in\{0.01,0.01,1,10\}$. Altogether, we run simulations  for 56 different parametrisations of SA and QMC-SA. The results presented in this subsection are obtained from 1\,000 starting values sampled independently and uniformly on $\setX_1$.

\begin{figure}
\centering

\begin{subfigure}{0.28\textwidth}
\centering
\includegraphics[trim=2cm 0 0 0, scale=0.23]{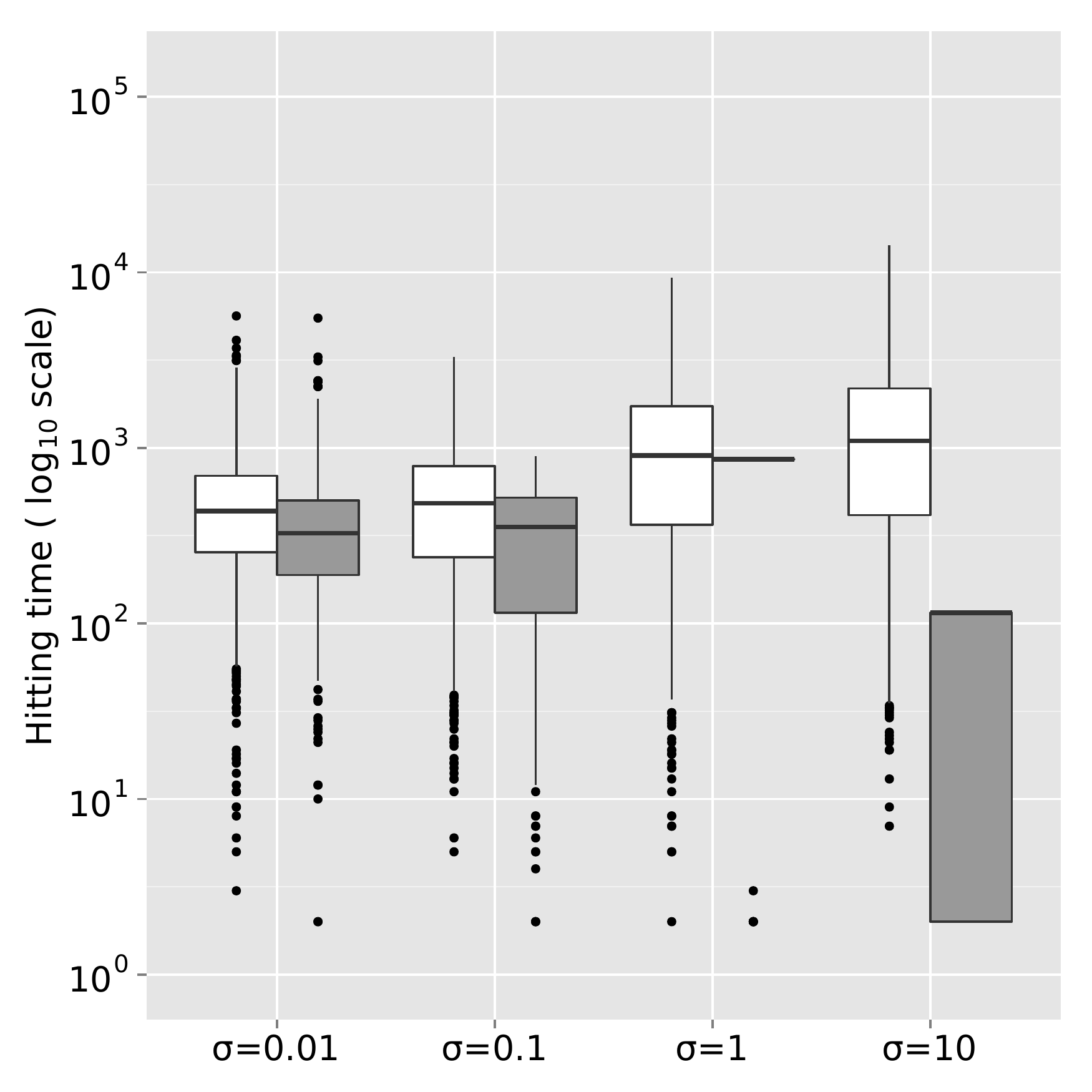}
\caption{\label{Cauchy_QMC2}}
\end{subfigure}
\hspace{0.5cm}
\begin{subfigure}{0.28\textwidth}
\includegraphics[trim=2cm 0 0 0, scale=0.23]{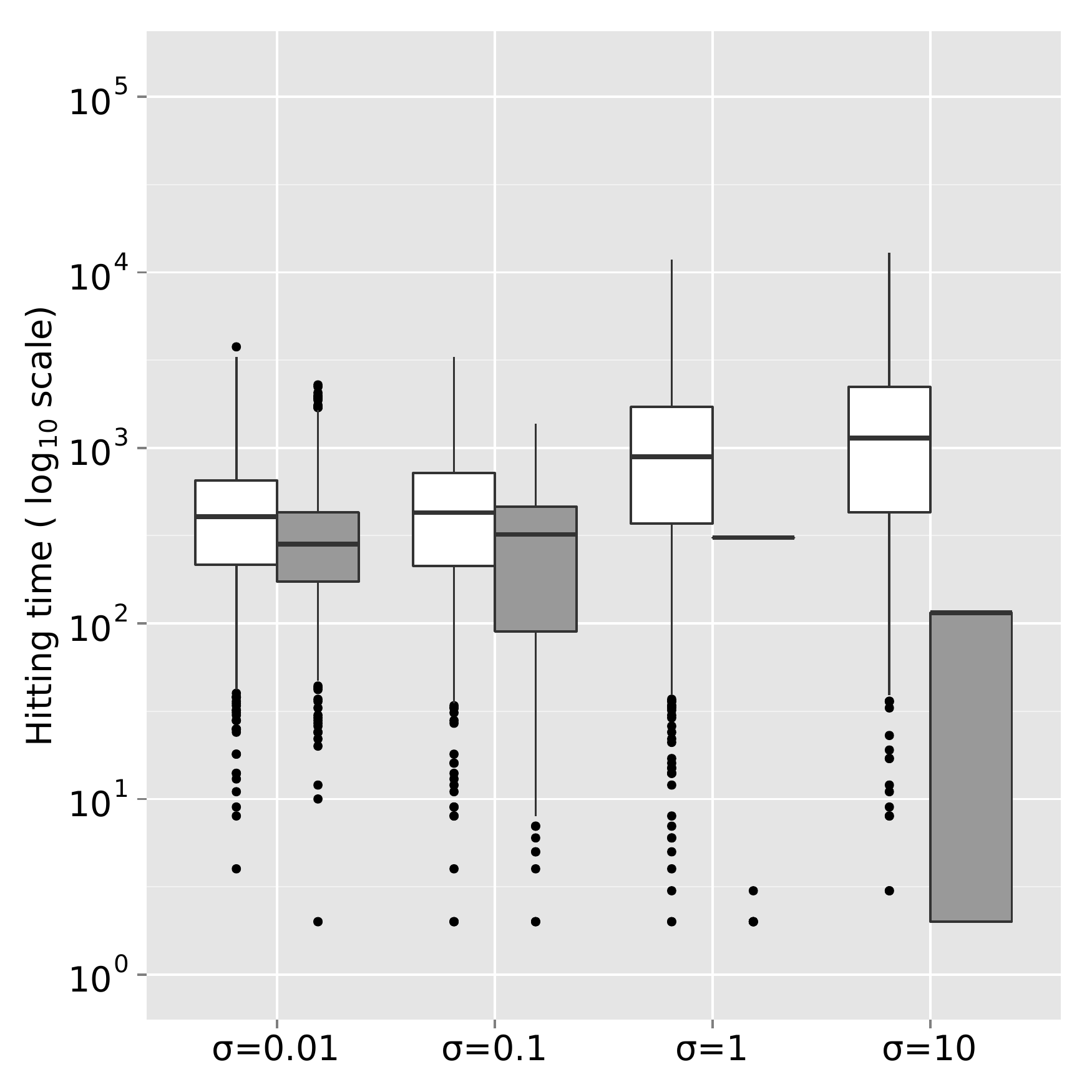}
\caption{\label{Cauchy_Lin2}}
\end{subfigure}
\hspace{0.5cm}
\begin{subfigure}{0.28\textwidth}
\includegraphics[trim=2cm 0 0 0, scale=0.23]{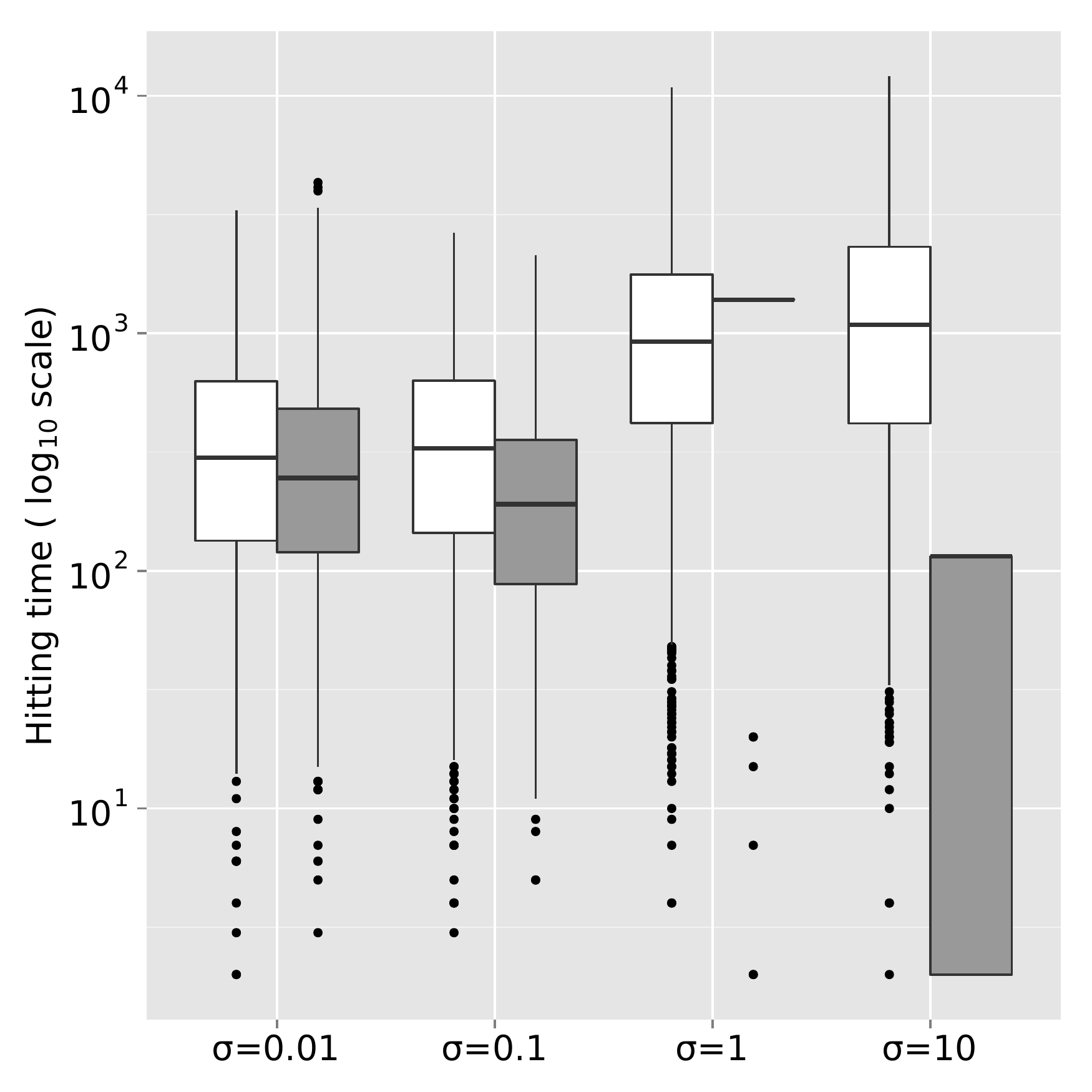}
\caption{\label{Cauchy_Log2}}
\end{subfigure}

\begin{subfigure}{0.28\textwidth}
\centering
\includegraphics[trim=2cm 0 0 0, scale=0.23]{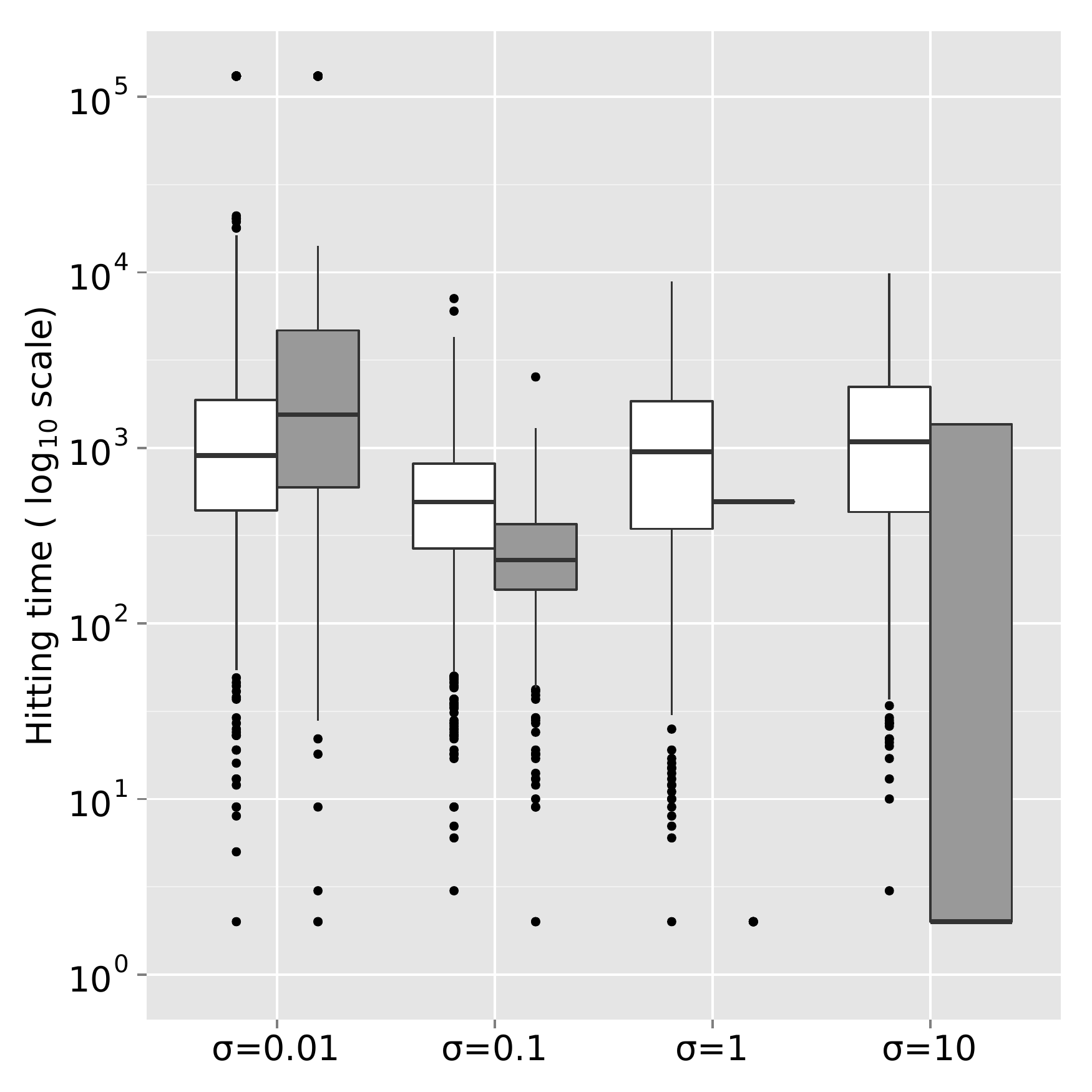}
\caption{\label{Gaussian_QMC2}}
\end{subfigure}
\hspace{0.5cm}
\begin{subfigure}{0.28\textwidth}
\includegraphics[trim=2cm 0 0 0, scale=0.23]{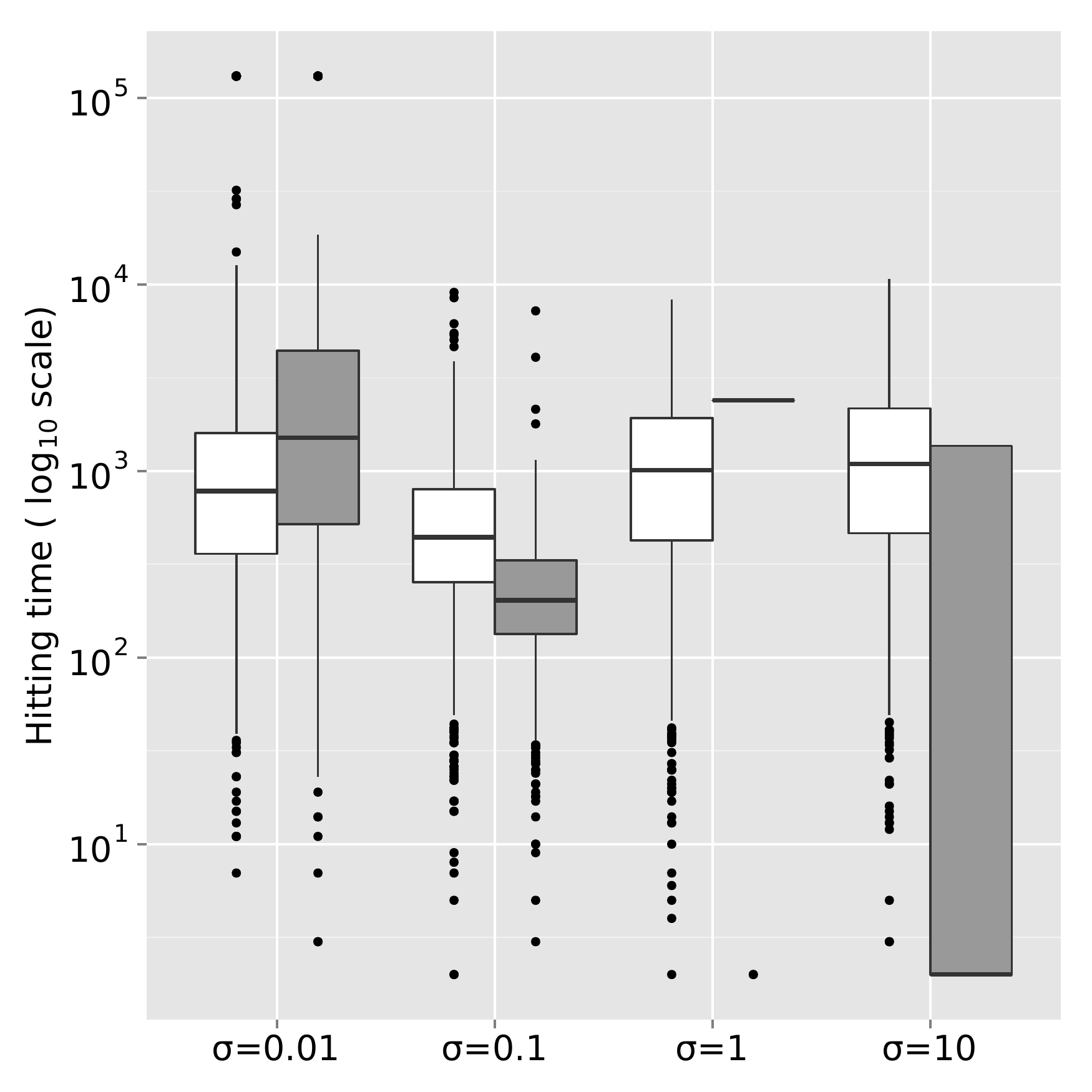}
\caption{\label{Gaussian_Lin2}}
\end{subfigure}
\hspace{0.5cm}
\begin{subfigure}{0.28\textwidth}
\includegraphics[trim=2cm 0 0 0, scale=0.23]{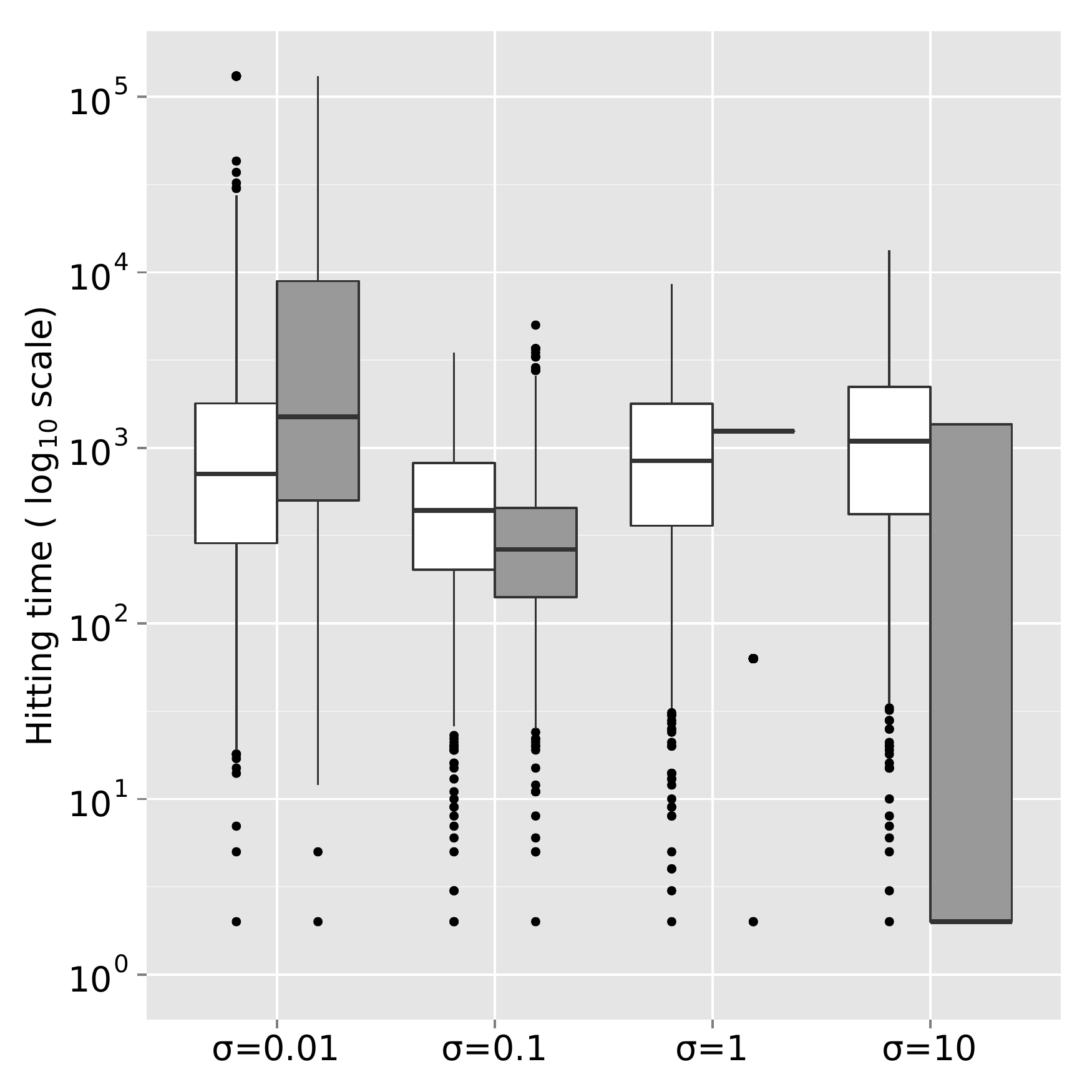}
\caption{\label{Gaussian_Log2}}
\end{subfigure}
\caption{Minimization of $\tilde{\varphi}_1$ defined by \eqref{eq:toy} for 1\,000  starting values sampled independently and uniformly on $\setX_1$. Results are presented for  for the Cauchy kernel (top)  and for the Gaussian kernel. For each kernel, simulation are done for $(T_n^{(m)})_{n\geq 1}$ where $m=1$ (left plots), $m=2$ (middle) and $m=3$, with $T_0^{(m)}$ the median of the values given in \eqref{num:Temp2}. The plots show the minimum number of iterations needed for SA (white boxes) and QMC-SA to find a $\bx\in \setX_1$ such that $\tilde{\varphi}_1(\bx)<10^{-5}$. For each starting value, the Monte Carlo algorithm is run only once and the QMC-SA algorithm is based on the Sobol' sequence.\label{fig:Robert}}
\end{figure}

Figure \ref{fig:Robert} shows the results for the two kernels and for the  sequences of temperatures given in \eqref{num:Temp} where, for $m\in 1:3$, $T_0^{(m)}$ is the median value given in \eqref{num:Temp2}. The results for the other  values of $T_0^{(m)}$ are presented in  Appendix \ref{app:fig} (Figures \ref{fig:Robert_Cauchy} and \ref{fig:Robert_Gaussian}).

Focussing first on the results for the Cauchy  kernel (first  row), we observe that QMC-SA is never unambiguously  worst than SA   and is significantly better in most cases. The performance of QMC-SA  becomes globally better as we increase the step size $\sigma$ (we however note that QMC-SA  tends to be the less efficient when $\sigma=1$) with the best results for QMC-SA  obtained when $\sigma=10$ where, in several settings, the maximum  hitting time of an error of size $10^{-5}$ is around 100 (against $10^{3.5}$ for SA).  A last interesting observation we can make from this first set of simulation  is that we obtain in several cases exactly the same hitting time for different starting values of the QMC-SA.

The results for the Gaussian kernel  show a slightly different story. Indeed, when $\sigma=0.01$, QMC-SA provides a poorer performance than SA. The reason for this is that, as the tails of the Gaussian distribution are very small in this case, a large value of $\|\bu_{\infty}^n\|_{\infty}$ is needed at iteration $n$ of the algorithm to generate a candidate value $\by^n$ far away from the current location. However, the number of such  points  are limited when we go through a $(t,d)$-sequence and therefore QMC-SA explores  the state space $\setX_1$ mostly through very small moves when a kernel with very tiny tails is used.  As a takeaway, when the proposal kernel has small variance alongside light tails, little benefit is obtained by QMC-SA.

Altogether, we observe that SA unambiguously outperforms QMC-SA in only 8 of the 56 scenarios under study (Figures \ref{Gaussian_QMC2}-\ref{Gaussian_Log2} with $\sigma=1$, and, in Appendix \ref{app:fig}, Figure \ref{Cauchy_Lin1} with $\sigma=1$, Figure \ref{Gaussian_QMC1} with $\sigma=1$ and  Figures \ref{Gaussian_QMC3}-\ref{Gaussian_Log3} with $\sigma=0.01$). However, none of these situations correspond to a good (and hence desirable) parametrization of SA.

\subsection{Example 2: Application to spatial statistics\label{sub:spatial}}

\begin{figure}
\centering

\begin{subfigure}{0.28\textwidth}
\includegraphics[trim=2cm 0 0 0, scale=0.23]{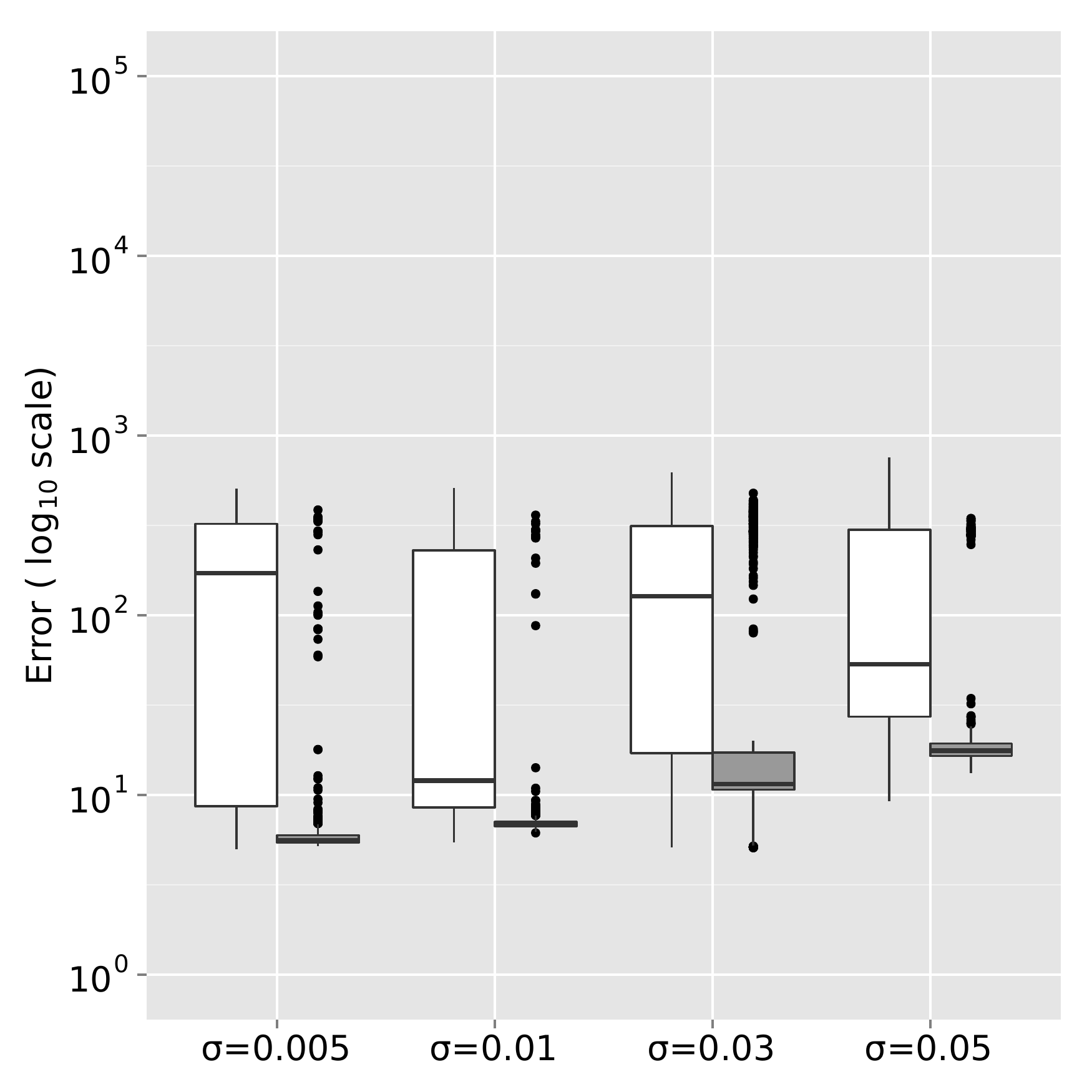}
\caption{\label{Fig:QMC2}}
\end{subfigure}
\hspace{1.2cm}
\begin{subfigure}{0.28\textwidth}
\includegraphics[trim=2cm 0 0 0, scale=0.23]{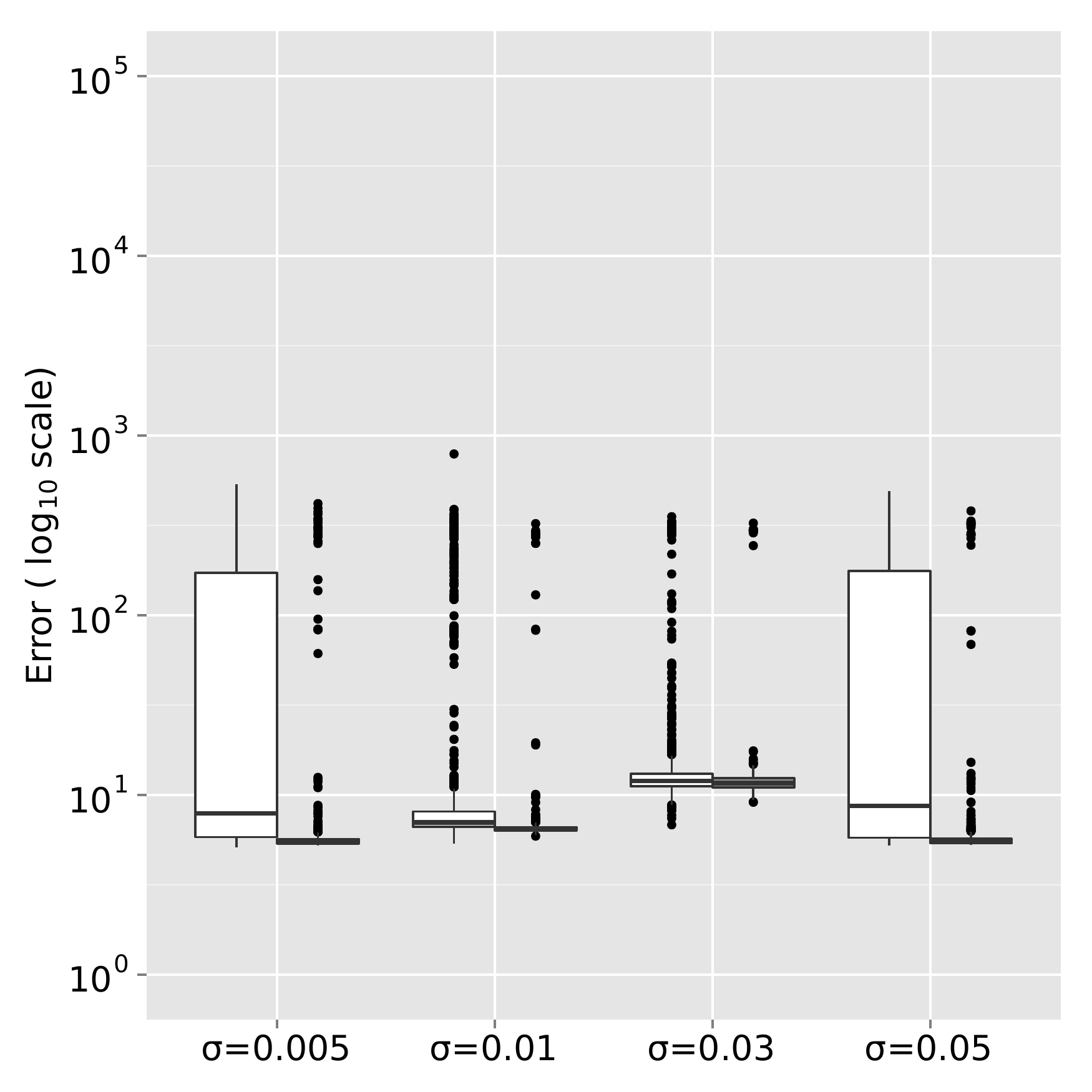}
\caption{\label{Fig:Log2}}
\end{subfigure}

\begin{subfigure}{0.28\textwidth}
\includegraphics[trim=2cm 0 0 0, scale=0.23]{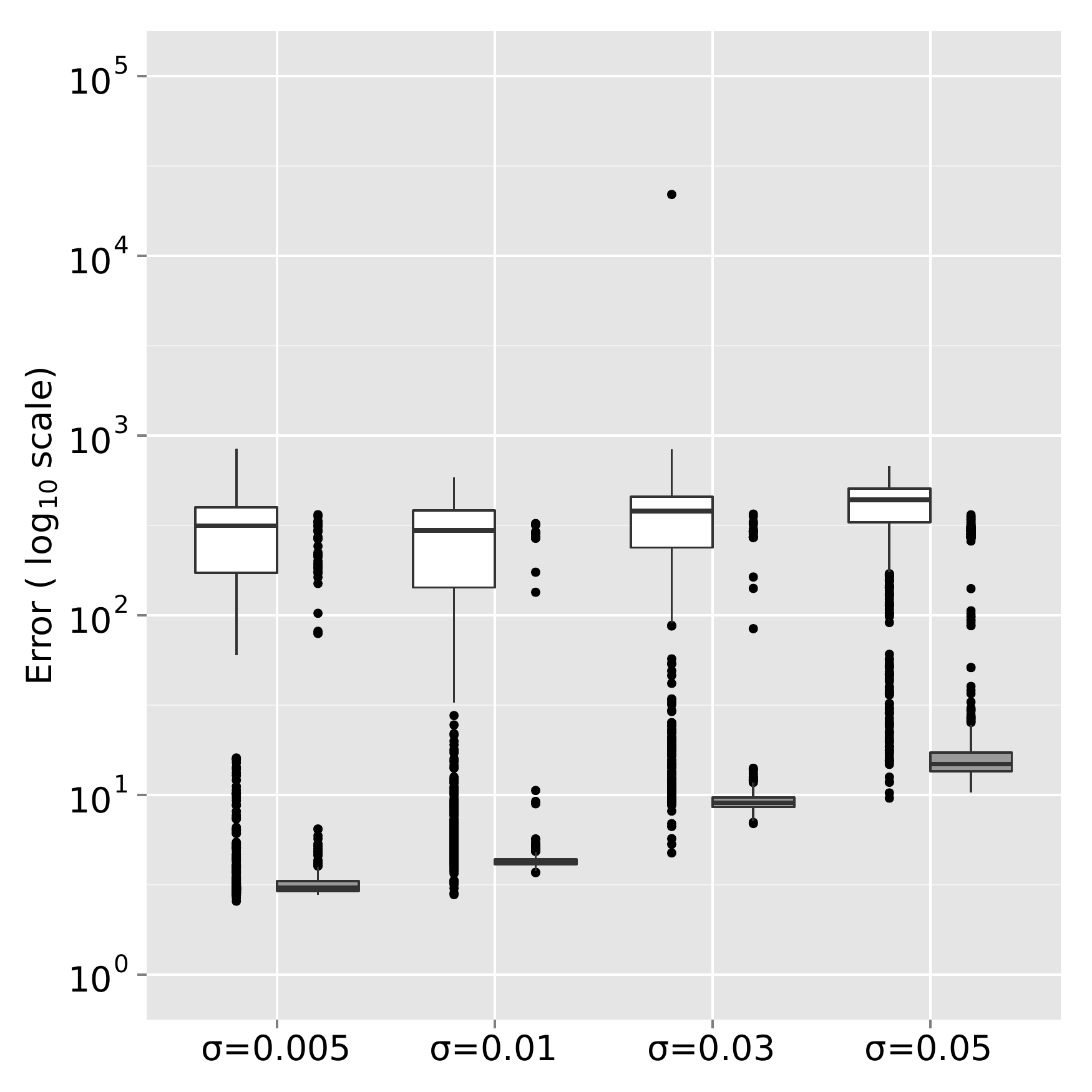}
\caption{\label{Fig:QMC1}}
\end{subfigure}
\hspace{1.2cm}
\begin{subfigure}{0.28\textwidth}
\includegraphics[trim=2cm 0 0 0, scale=0.23]{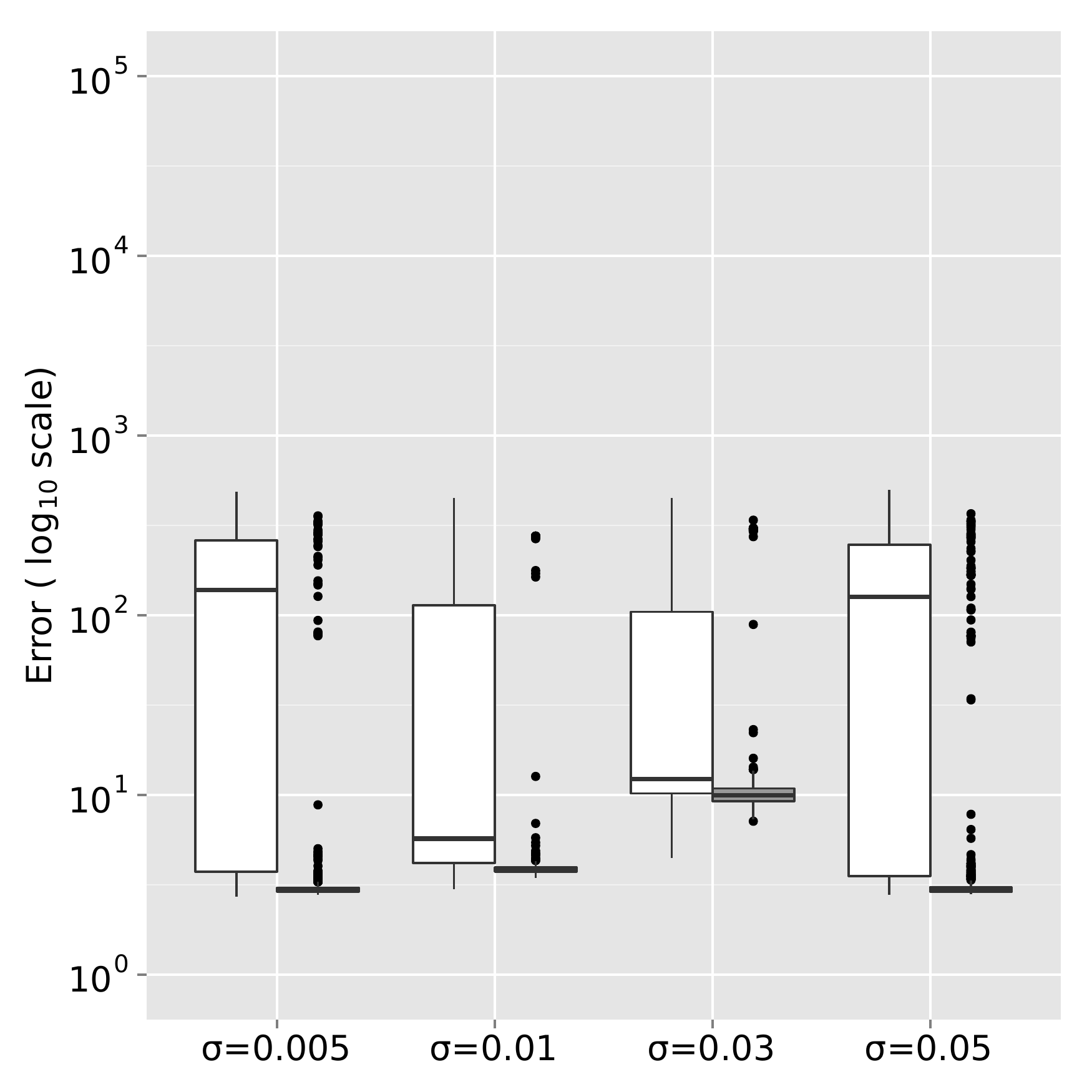}
\caption{\label{Fig:Log1}}
\end{subfigure}

\caption{Minimization of $\tilde{\varphi}_{\lambda}$  for 1\,000 starting values sampled independently in $\setX_2$ (as explained in the text) and for $\lambda=0.1$ (top) and for $\lambda=0.01$ (bottom). Results are presented for a Cauchy random walk with step size as described in the text, and for $(T_n^{(1)})_{n\geq 1}$ with $T_0^{(1)}=5\,000$ (left) and for $(T_n^{(4)})_{n\geq 1}$ with $T_0^{(4)}=0.01$ (right). The plots show $\min\{\tilde{\varphi}_2(\bx^n),\, n\in 1:2^{17}\}$ obtained by SA (white boxes) and QMC-SA for each starting value. The results are obtained for $d_1=100$ locations. For each starting value  the Monte Carlo algorithm is run only once and the QMC-SA algorithm is based on the Sobol' sequence.\label{Fig:Dim exp}}
\end{figure}

Let $\{Y(\bx):\,\bx\in\mathbb{R}^{2}\}$ be a spatial process and consider the problem of estimating the variogram $\E(|Y(\bx_i)-Y(\bx_j)|^2)$ for $i,j\in 1:d_1$. When the process is assumed to be stationary, this estimation is typically straightforward. However, for most real-world spatial problems arising in climatology, environmetrics, and elsewhere, inference is much more challenging as the underlying process $\{Y(\bx):\,\bx\in\mathbb{R}^{2}\}$ is inherently nonstationary.

A simple  way to modeling nonstationary spatial processes is to use a dimension expansion approach, as proposed by \citet{Bornn2012}. For the sake of simplicity, we assume that the process $\{Y([\bx,z]):\,[\bx,z]\in\mathbb{R}^{3}\}$ is stationary; that is, adding only one dimension is enough to get a stationary process. Thus, the variogram of this process
depends only on the distance between location $i$ and $j$ and can be modelled, e.g., using the parametric model
$$
\gamma_{\phi_1,\phi_2}([\bx_1,z_1],[\bx_2,z_2])=\phi_1(1-\exp\{-\|[\bx_1,z_1]-[\bx_2,z_2]\|/\phi_2\})
$$
where $\phi_1$ and $\phi_2$ are two positive parameters to be learned from the data.

Assuming that we have $M\geq 2$ observations $\{y_{m,i}\}_{m=1}^M$ at  location $i\in 1: d_1$, the solution to this problem is obtained by minimizing $\tilde{\varphi}_{\lambda}:\setX_2:=\mathbb{R}^{+}\times \mathbb{R}^{+}\times \mathbb{R}^{d_1}\rightarrow\mathbb{R}^+$, defined by
\begin{equation}\label{num:phi2}
\begin{split}
\tilde{\varphi}_{\lambda}(\phi_1,\phi_2,\bz)&=\sum_{1\leq i<j}^{d_1} \Big\{v_M^{(ij)}-\gamma_{\phi_1,\phi_2}([\bx_i,z_i],[\bx_j,z_j])\Big\}^2+\lambda \|\bz\|_1
\end{split}
\end{equation}
where  $\lambda>0$ control the regularity of $\bz$ and where $v_M^{(ij)}=M^{-2}\sum_{m,m'=1}^M |y_{m,i}-y_{m',j}|^2$
is an estimate of the spatial dispersion between locations $i$ and $j$. Note that  $\tilde{\varphi}_{\lambda}$ is non-differentiable  because of the $L_1$-penalty and is both high-dimensional (with $1$ parameter per observation) as well as nonlinear (due to way the latent locations factor into the parametric variogram). To further complicate matters, the objective function's minimum is only unique up to rotation and scaling of the latent locations. 

Following \citet{Bornn2012}, the  observations are generated by simulating a Gaussian process, with  $d_1=100$ locations on a three dimensional half-ellipsoid centered at  the origin and $M=1\,000$. We minimize  the objective function \eqref{num:phi2} using both SA and QMC-SA with a Cauchy random walk defined by
$$
K(\bx,\dd\by)=\left(\otimes_{i=1}^2 f^{(1)}_{\mathbb{R}^+}(y_i,x_i,\sigma^{(i)})\dd y_i\right)\otimes\left(\otimes_{i=3}^d f^{(1)}_{\mathbb{R}}(y_i,x_i,\sigma^{(i)})\dd y_i\right)
$$
where $f^{(1)}_{I}(\cdot,\mu,\tilde{\sigma})$ is as in the previous subsection. Simulations are performed for
$$
(\sigma^{(1)},\dots,\sigma^{(d)})=\sigma\times\big(0.1,0.1, 0.5\,\widehat{\sigma}_M(y_{\bm{\cdot},1}),\dots, 0.5\widehat{\sigma}_M(y_{\bm{\cdot},d_1})\big),
$$
where $\widehat{\sigma}_M(y_{\bm{\cdot},i})$ denotes the  standard deviation of  the observations $\{y_{m,i}\}_{m=1}^M$ at location $i\in 1:d_1$. Simulations are conducted for $\sigma\in\{0.005,0.01,0.03,0.05\}$ and for the sequence of temperatures  $(T_n^{(1)})_{n\geq 1}$ given in \eqref{num:Temp} and for $(T_n^{(4)})_{n\geq 1}$ defined by $T^{(4)}_n=T^{(4)}_0/\log(n+C^{(4)})$. As already mentioned (see Section \ref{sub:consisistency}), the sequence $(T_n^{(4)})_{n\geq 1}$ is such that convergence results for SA on unbounded spaces exist but the constants $T_0^{(4)}$ and $C^{(4)}$ are model dependent and intractable. In this simulation study, we take $T_0^{(1)}=5\,000$, $T_0^{(4)}=0.1$ and $C^{(4)}=100$. These values are chosen based on  some pilot runs of SA and QMC-SA fo $\lambda=0.1$.  Note also that the values of $T_0^{(1)}$, $T_0^{(4)}$ and $C^{(4)}$ are such that   $T_N^{(1)}\approx T_N^{(4)}$, where  $N=2^{17}$ is the number of function evaluations we consider in this numerical study.

Figures \ref{Fig:QMC2}-\ref{Fig:Log2} shows the values of $\min\{\tilde{\varphi}_{\lambda}(\bx^n),\, n\in 0: N\}$ obtained by SA and QMC-SA when for $\lambda=0.1$ and for 1\,000 different starting values  
$\bx_0=(\varphi_{0,1:2}, z_{0,1:d_1})$ sampled independently in $\setX_2$ as follows:
$$
\varphi_{0,1}\sim\Unif (0,2),\quad \varphi_{0,2}\sim\Unif (0,2),\quad z_{0,i}\sim\mathcal{N}(0,1),\quad i\in 1:d_1.
$$

Looking first at Figure \ref{Fig:Log2}, we observe  that SA performs relatively well when  the sequence $(T_n^{(4)})_{n\geq 1}$ is used (together with $\sigma\in\{0.01,0.03\}$). The good performance of SA in this case is not surprising since $T_0^{(4)}$ and $C^{(4)}$ are calibrated such that SA works well when $\lambda=0.1$. However, we remark that, for this sequence of temperatures, QMC-SA  outperforms SA for any value of $\sigma$. 
The results obtained for the sequence $(T_n^{(1)})_{n\geq 1}$ are presented in Figure \ref{Fig:QMC2}. If in this case SA performs poorly, and as for the sequence $(T_n^{(4)})_{n\geq 1}$,  QMC-SA provides a small error for the  vast majority of the 1\,000 different staring values (in particular when $\sigma\in\{0.005,0.01\}$).

In practice, one often minimizes $\tilde{\varphi}_{\lambda}$  for different values of the $\lambda$ parameter, which determines the regularity of the optimal solution for  $\bz\in\mathbb{R}^{d_1}$. It is therefore important that the optimization method at hand remains efficient for different values of $\lambda$. To evaluate the sensitivity of QMC-SA and SA to this parameter,  
Figures \ref{Fig:QMC1}-\ref{Fig:Log1} show the results obtained $\lambda=0.01$. For both sequences of temperatures, we observe that SA is much less efficient than when $\lambda=0.1$. In particular,  the results for the sequence $(T_n^{(4)})_{n\geq 1}$ (Figure \ref{Fig:Log1}) suggest that  SA is very sensitive to the cooling schedule in this example. On the contrary,  QMC-SA performs well in all cases and
outperforms (from far) SA for the two sequences of temperatures and for all values of $\sigma$ we have chosen for this numerical study.

The main message of this example is that the performance of QMC-SA on unbounded spaces is very robust to different choice of step-size $\sigma$ and of the cooling schedule. Consequently, tuning QMC-SA is much  simpler than for SA algorithms. The results of this subsection seem to indicate that the convergence of QMC-SA on unbounded spaces could be obtained under the same condition for $(T_n)_{n\geq 1}$ than for compact spaces (see Theorem \ref{thm:conv}). In particular, this suggests that there exists a universal sequence of cooling schedules which ensure the convergence of QMC-SA on non-compact space. However, further research in this direction is needed.

\section{Conclusion\label{sec:conc}}

In this paper we show that the performance of simulated annealing algorithms can be greatly improved through derandomization.  Indeed,  in the extensive numerical study proposed in this work we never observe a situation where SA performs well while QMC-SA performs poorly. In addition, in the vast majority of the scenarios under study, QMC-SA turns out to be much better than plain Monte Carlo SA. This is particularly true in the high dimensional example proposed in this work where in most cases plain SA fails to provide a satisfactory solution to the optimization problem, while QMC-SA does.

Our theoretical results also advance the current understanding of simulated annealing and related algorithms, demonstrating almost sure convergence with no objective-dependent conditions on the cooling schedule. These results also hold for classical SA under minimal assumptions on the objective function. Further, the convergence results extend beyond SA to a broader class of optimization algorithms including threshold accepting.

Future research should study  QMC-SA on non compact spaces and extend QMC-SA to other types of Markov kernels. Concerning this first point, it is of great interest to check if the consistency of QMC-SA on unbounded spaces can, as it is the case for compact spaces, be guaranteed for a universal (i.e. not model dependent) sequence of temperatures. Our simulation results suggest that this is indeed the case and therefore should encourage research in this direction.


In many practical optimization problems, and in particular in those arising in statistics, the objective function is multimodal but differentiable. In that case, stochastic gradient search algorithms \citep[see, e.g.,][]{Gelfand1991,Gelfand1993} are efficient alternatives to SA which, informally, correspond to SA  based on a Langevin type Markov transitions. Intuitively, combining the information on the objective function  provided by its gradient  with the good equidistribution properties of QMC point sets should result in a powerful search algorithm. We leave this exciting question for future research.

 \section*{Acknowledgement}

The authors acknowledge support from DARPA under Grant No. FA8750-14-2-0117. 
The authors also thank Christophe Andrieu, Pierre  Jacob and Art  Owen for insightful discussions and useful feedback.

\appendix

\section{Proof of Lemma \ref{lem:dense}\label{p-lem:dense}}

Let $n\in\mathbb{N}$, $(\tilde{\bx},\bx')\in\setX^2$, $\delta_{\setX}=0.5$ and   $\delta\in(0,\delta_{\setX}]$. Then, by Assumption \ref{H:lem:K1}, $F_K^{-1}(\tilde{\bx},\bu_1^n)\in B_{\delta}(\bx')$  if and only if $
\bu_1^n\in F_K(\tilde{\bx},B_{\delta}(\bx'))$. We now  show that, for $\delta$ small enough, there exists a closed hypercube $W(\tilde{\bx},\bx', \delta)\subset\ui^d$  such that $W(\tilde{\bx},\bx',\delta)\subseteq F_K(\bx,B_{\delta}(\bx'))$ for all $\bx\in B_{v_K(\delta)}(\tilde{\bx})$, with $v_K(\cdot)$  as in the statement of the lemma.

To see this note that, because $K(\bx,\dd\by)$ admits a density $K(\by|\bx)$ which is continuous on the compact set $\setX^2$, and using Assumption \ref{H:lem:K2}, it is easy to see that, for $i\in 1:d$, $K_i(y_i|y_{1:i-1},\bx)\geq\tilde{K}$    for all $(\bx,\by)\in\setX^2$ and for a constant $\tilde{K}>0$.  Consequently, for any  $\delta\in [0,0.5]$  and $(\bx,\by)\in\setX^2$,
\begin{align}\label{eq:p-lem:denseB1}
F_{K_i}\left(x'_i+\delta|\bx,  y_{1:i-1}\right)-F_{K_i}\left(x'_i-\delta|\bx,  y_{1:i-1}\right)\geq \tilde{K}\delta,\quad \forall i\in 1:d
\end{align}
where $F_{K_i}(\cdot{}|\bx, y_{1:i-1})$ denotes the CDF of  the probability measure $K_i(\bx,y_{1:i-1},\dd y_i)$, with the convention that $F_{K_i}(\cdot{}|\bx, y_{1:i-1})=F_{K_1}(\cdot{}|\bx)$ when $i=1$.
Note that the right-hand side of \eqref{eq:p-lem:denseB1} is $\tilde{K}\delta$ and not $2\tilde{K}\delta$ to  encompass the case where either $x_i'-\delta\not\in[0,1]$ or  $x_i'+\delta\not\in[0,1]$. (Note also that because $\delta\leq 0.5$ we cannot have both  $x_i'-\delta\not\in[0,1]$ and  $x_i'+\delta\not\in[0,1]$.)

For $i\in 1:d$ and  $\delta'>0$, let
$$
\omega_i(\delta')=\sup_{
\substack{(\bx,\by)\in\setX^2,\,(\bx',\by')\in\setX^2\\
\|\bx-\bx'\|_{\infty}\vee\|\by-\by'\|_{\infty}\leq\delta'
}}|F_{K_i}\left(y_i|\bx, y_{1:i-1}\right)-F_{K_i}\left(y_i'|\bx',  y'_{1:i-1}\right)|
$$
be the (optimal) modulus of continuity of $F_{K_i}(\cdot|\cdot)$. Since $F_{K_i}$ is uniformly continuous on the compact  set  $[0,1]^{d+i}$, the  mapping $\omega_i(\cdot)$ is continuous   and $\omega_i(\delta')\cvz $ as $\delta'\cvz$. In addition, because $F_{K_i}\left(\cdot|\bx, y_{1:i-1}\right)$ is strictly increasing  on $[0,1]$ for all $(\bx,\by)\in\setX^2$, $\omega_i(\cdot)$ is strictly increasing on $(0,1]$. Let $\tilde{K}$ be small enough so that, for $i\in 1:d$, $0.25\tilde{K}\delta_{\setX}\leq w_i(1)$ and let  $\tilde{\delta}_i(\cdot)$ be the mapping $
z\in (0,\delta_{\setX}]\longmapsto\tilde{\delta}_i(z)=\omega_i^{-1} (0.25\tilde{K}z)$. Remark that the function $\tilde{\delta}_i(\cdot)$ is independent of $(\tilde{\bx},\bx')\in\setX^2$, continuous and strictly increasing on $(0,\delta_{\setX}]$ and such that $\tilde{\delta}_i(\delta')\cvz$ as $\delta'\cvz$.

For $\bx\in\setX$, $\delta'>0$ and $\delta'_i>0$, $i\in 1:d$, let
$$
B^i_{\delta'}(\tilde{\bx})=\{\bx\in [0,1]^i:\, \|\bx-\tilde{x}_{1:i}\|_{\infty}\leq \delta'\}\cap [0,1]^i
$$
and
$$
B_{\delta'_{1:i}}(\tilde{\bx})=\{\bx\in [0,1]^i:\, |x_j-\tilde{x}_{j}|\leq \delta'_j,\,j\in 1:i\}\cap [0,1]^i.
$$

Then, for any $\delta'>0$ and for all $(\bx,y_{1:i-1})\in B_{\tilde{\delta}_i(\delta)}(\tilde{\bx})\times  B^{i-1}_{\tilde{\delta}_i(\delta)}(\bx')$, we have
\begin{align}
&|F_{K_i}\left( x'_i+\delta'|\bx,  y_{1:i-1}\right)-F_{K_i}\left(x'_i+\delta'|\tilde{\bx},  x'_{1:i-1}\right)|\leq  0.25\tilde{K}\delta\label{eq:p-lem:denseB2}\\
&|F_{K_i}\left(x'_i-\delta'|\bx,  y_{1:i-1}\right)-F_{K_i}\left(x'_i-\delta'|\tilde{\bx},  x'_{1:i-1}\right)|\leq  0.25\tilde{K}\delta\label{eq:p-lem:denseB3}.
\end{align}

For $i\in 1:d$ and $\delta'\in (0,\delta_{\setX}]$, let  $\delta_i(\delta')=\tilde{\delta}_i(\delta')\wedge\delta'$ and note that the function $\delta_i(\cdot)$ is continuous and strictly increasing on $(0,\delta_{\setX}]$. Let $\delta_d=\delta_d(\delta)$ and define recursively $\delta_{i}=\delta_{i}(\delta_{i+1})$, $i\in 1:(d-1)$, so that $\delta\geq \delta_d\geq\dots\geq\delta_1>0$. For $i\in 1:d$, let
$$
\underline{v}_i(\tilde{\bx},\bx',\delta_{1:i})= \sup_{(\bx,y_{1:i-1})\in B_{\delta_{1}}(\tilde{\bx})\times B_{\delta_{1:i-1}}(\bx')} F_{K_i}\left(x'_i-\delta_i|\bx, y_{1:i-1}\right)
$$
and
$$
\bar{v}_i(\tilde{\bx},\bx',\delta_{1:i})= \inf_{(\bx,y_{1:i-1})\in B_{\delta_{1}}(\tilde{\bx})\times B_{\delta_{1:i-1}}(\bx')} F_{K_i}\left(x'_i+\delta_i|\bx, y_{1:i-1}\right).
$$
Then, since $F_{K_i}(\cdot|\cdot)$ is continuous and the set $B_{\delta_{1}}(\tilde{\bx})\times B_{\delta_{1:i-1}}(\bx')$ is compact, there exists points $(\underline{\bx}^i,\underline{y}^{i}_{1:i-1})$ and $(\bar{\bx}^i,\bar{y}^{i}_{1:i-1})$ in $B_{\delta_{1}}(\tilde{\bx})\times B_{\delta_{1:i-1}}(\bx')$ such that
$$
\underline{v}_i(\tilde{\bx},\bx',\delta_{1:i})=F_{K_i}(x'_i-\delta_{i}|\underline{\bx}^i, \underline{y}^{i}_{1:i-1}),\quad \bar{v}_i(\tilde{\bx},\bx', \delta_{1:i})=F_{K_i}\left(x'_i+\delta_i|\bar{\bx}^i, \bar{y}^{i}_{1:i-1}\right).
$$
In addition, by the construction of the $\delta_i$'s, $B_{\delta_{1}}(\tilde{\bx})\times B_{\delta_{1:i-1}}(\bx')\subseteq B_{\tilde{\delta}_i(\delta_{i})}(\tilde{\bx})\times B^i_{\tilde{\delta}_i(\delta_{i})}(\bx')$ for all $i\in 1:d$. Therefore,  using \eqref{eq:p-lem:denseB1}-\eqref{eq:p-lem:denseB3}, we have, for all $i\in 1:d$,
\begin{align*}
\bar{v}_i(\tilde{\bx},\bx',\delta_{1:i})-\underline{v}_i(\tilde{\bx},\bx',\delta_{1:i})&=F_{K_i}(x'_i+\delta_i|\bar{\bx}^i,  \bar{y}^i_{1:i-1})-F_{K_i}(x'_i-\delta_i|\underline{\bx}^i,  \underline{y}^i_{1:i-1})\\
&\geq F_{K_i}\left(x'_i+\delta_i|\tilde{\bx},  x'_{1:i-1}\right)-F_{K_i}\left(x'_i-\delta_i|\tilde{\bx},  x'_{1:i-1}\right)-0.5\tilde{K}\delta_i\\
&\geq 0.5\tilde{K}\delta_i\\
&>0.
\end{align*}
Consequently, for all $i\in 1:d$ and for all $(\bx,y_{1:i-1})\in B_{\delta_{1}}(\tilde{\bx})\times  B_{\delta_{1:i-1}}(\bx')$,
$$
\Big[\underline{v}_i(\tilde{\bx},\bx',\delta_{1:i}), \bar{v}_i(\tilde{\bx},\bx',\delta_{1:i})\Big]\subseteq \Big[F_{K_i}(x'_i-\delta_i|\bx, y_{1:i-1}),F_{K_i}(x'_i+\delta_i|\bx, y_{1:i-1})\Big].
$$
Let $\underline{S}_{\delta}=0.5\tilde{K}\delta_1$. Then, this shows that there exists a closed hypercube $\underline{W}(\tilde{\bx},\bx',\delta)$ of side $\underline{S}_{\delta}$
such that
$$
\underline{W}(\tilde{\bx},\bx',\delta)\subseteq F_K\big(\bx,B_{\delta_{1:d}}(\bx')\big)\subseteq F_K\big(\bx,B_{\delta}(\bx')\big),\quad\forall\bx\in B_{v_K(\delta)}(\tilde{\bx})
$$
where we set   $v_K(\delta)=\delta_1$. Note that $v_K(\delta)\in (0,\delta]$ and thus $v_K(\delta)\cvz$ as $\delta\cvz$, as required. In addition, $v_K(\cdot)=\delta_1\circ\dots\delta_d(\cdot)$ is continuous and strictly increasing on $(0,\delta_{\setX}]$ because  the functions $\delta_i(\cdot)$, $i\in 1:d$, are  continuous and strictly increasing on this set. Note also that $v_K(\cdot)$ does not depend on $(\tilde{\bx},\bx')\in\setX^2$.

To conclude the proof, let
\begin{align}\label{eq:p-lem:denseB1_kdelta}
k_{\delta}=\big\lceil t+d-d\log(\underline{S}_{\delta}/3)/\log b\big\rceil
\end{align}
and note that, if $\delta$ is small enough,  $k_{\delta}\geq t+d$ because $\underline{S}_{\delta}\cvz$ as $\delta\cvz$. Let $\bar{\delta}_K$ be the largest value of $\delta'\leq\delta_{\setX}$ such that $k_{\delta'}\geq t+d$. Let $\delta\in(0,\bar{\delta}_K]$ and $t_{\delta,d}\in t:(t+d)$ be such that $(k_{\delta}-t_{\delta,d})/d$ is an integer.  Let $\{E(j,\delta)\}_{j=1}^{b^{k_{\delta}-t_{\delta,d}}}$ be the partition of  $\ui^d$ into elementary intervals of volume $b^{t_{\delta,d}-k_{\delta}}$ so that any closed hypercube of side $\underline{S}_{\delta}$ contains at least one elementary interval $E(j,\delta)$ for a $j\in 1:b^{k_{\delta}-t_{\delta,d}}$. Hence, there exists a $j_{\tilde{\bx},\bx'}\in 1:b^{k_{\delta}-t_{\delta,d}}$ such that 
$$
E(j_{\tilde{\bx},\bx',},\delta)\subseteq \underline{W}(\tilde{\bx},\bx',\delta)\subseteq F_K(\bx,B_{\delta}(\bx')),\quad\forall\bx\in B_{v(\delta)}(\tilde{\bx}).
$$
Let $a\in\mathbb{N}$ and note that, by the properties of $(t,s)$-sequences in base $b$,  the point set $\{\bu^n\}_{n=ab^{k_{\delta}}}^{(a+1)b^{k_{\delta}}-1}$ is a $(t,k_{\delta},d)$-net in base $b$ because $k_{\delta}> t$. In addition, since $k_{\delta}\geq t_{\delta,d}\geq t$, the point set $\{\bu^n\}_{n=ab^{k_{\delta}}}^{(a+1)b^{k_{\delta}}-1}$ is also a $(t_{\delta,d},k_{\delta},d)$-net in base $b$ \citep[][Remark 4.3, p.48]{Niederreiter1992}. Thus, since for $j\in  1:b^{k_{\delta}-t_{\delta,d}}$ the elementary interval $E(j,\delta)$ has volume $b^{t_{\delta,d}-k_{\delta}}$, the point set $\{\bu^n\}_{n=ab^{k_{\delta}}}^{(a+1)b^{k_{\delta}}-1}$ therefore contains exactly $b^{t_{\delta_d}}\geq b^t$  points in $E(j_{\tilde{\bx},\bx',},\delta)$ and the proof is complete.

\section{Proof of Lemma \ref{lem:dense2}\label{p-lem:dense2}}

Using the Lipschitz property of $F_{K_i}(\cdot|\cdot)$ for all $i\in 1:d$, conditions \eqref{eq:p-lem:denseB2} and \eqref{eq:p-lem:denseB3} in the proof of Lemma \ref{lem:dense} hold with $
\tilde{\delta}_i(\delta)=\delta (0.25\tilde{K}/C_K)$, $i\in 1:d$. Hence,  we can take $v_K(\delta)=\delta (0.25\tilde{K}/C_K)^{d}\wedge \delta$ and thus
$\underline{S}_{\delta}= \delta 0.5\tilde{K}\big(1 \wedge (0.25\tilde{K}/C_K)^{d}\big)$.
Then, the expression for $k_{\delta}$ follows using \eqref{eq:p-lem:denseB1_kdelta} while the expression for $\bar{\delta}_{K}\leq 0.5$ results from the condition $k_{\delta}\geq t+d$ for all $\delta\in (0,\bar{\delta}_{K}]$.

\section{Proof of Lemma \ref{lem:convPhi}}\label{p-lem:convPhi}

We first state and prove three technical lemmas:

\begin{lemma}\label{lem:upBound}
Let $\setX=[0,1]^d$ and $K:\setX\rightarrow\mathcal{P}(\setX)$ be a Markov kernel which verifies Assumptions \ref{H:lem:K1}-\ref{H:lem:K2}. Then, for any $\delta\in (0,\bar{\delta}_K]$, with $\bar{\delta}_K$ as in Lemma \ref{lem:dense}, and any $(\tilde{\bx},\bx')\in\setX^2$, there exists a closed hypercube $\bar{W}(\tilde{\bx},\bx',\delta)\subset\ui^d$ of side $\bar{S}_{\delta}=2.5\bar{K}\delta$, with $\bar{K}=\max_{i\in 1:d}\{\sup_{\bx,\by\in \setX}K_i(y_i|y_{1:i-1},\bx)\}$,  such that
\begin{align}\label{eq:p-lem:convPhi_E1}
F_K(\bx,B_{v_K(\delta)}(\bx'))\subseteq  \bar{W}(\tilde{\bx},\bx',\delta),\quad \forall \bx\in B_{v_K(\delta)}(\tilde{\bx})
\end{align}
where $v_K(\cdot)$ is as in Lemma \ref{lem:dense}.
\end{lemma}

\begin{proof}
The proof of Lemma \ref{lem:upBound} is similar to the proof of Lemma \ref{lem:dense}. Below,  we use the same notation as in this latter.

 Let $\delta\in (0,\bar{\delta}_K]$, $(\tilde{\bx},\,\bx')\in\setX^2$ and note that, for any $(\bx,\by)\in\setX^2$, 
\begin{align}\label{eq:p-lem:convPhiB11}
F_{K_i}(x'_i+\delta|\bx,y_{1:i-1})-F_{K_i}(x'_i-\delta|\bx,y_{1:i-1})\leq 2\bar{K}\delta,\quad i\in 1:d.
\end{align}
Let $0<\delta_1\leq\dots\leq\delta_d\leq \delta$  be as in the proof of Lemma \ref{lem:dense} and, for $i\in 1:d$, define 
$$
\underline{u}_i(\tilde{\bx},\bx',\delta_{1:i})=\inf_{(\bx,\,\by)\in B_{v_K(\delta)}(\tilde{\bx}),\times B_{\delta_{1:i-1}}(\bx')}F_{K_i}(x'_i-\delta_i|\bx,y_{1:i-1})
$$
and
$$
\bar{u}_i(\tilde{\bx},\bx', \delta_{1:i})=\sup_{(\bx,\,\by)\in B_{v_K(\delta)}(\tilde{\bx}),\times B_{\delta_{1:i-1}}(\bx')}F_{K_i}(x'_i+\delta_i|\bx,y_{1:i-1}).
$$

Let $i\in 1:d$ and $(\underline{\bx}^i,\underline{\by}^i),(\bar{\bx}^i,\bar{\by}^i)\in B_{v_K(\delta)}(\tilde{\bx})\times B_{\delta_{1:i-1}}(\bx')$  be such that
$$
\underline{u}_i(\tilde{\bx},\bx',\delta_{1:i})=F_{K_i}(x'_i-\delta_i|\underline{\bx}^i,\underline{y}^i_{1:i-1}),\quad \bar{u}_i(\tilde{\bx},\bx',\delta)=F_{K_i}(x'_i+\delta_i|\bar{\bx}^i,\bar{y}^i_{1:i-1}).
$$
Therefore, using \eqref{eq:p-lem:denseB2}, \eqref{eq:p-lem:denseB3} and \eqref{eq:p-lem:convPhi_E1}, we have, $\forall i\in 1:d$,
\begin{align*}
0<\bar{u}_i(\tilde{\bx},\bx',\delta_{1:i})-\underline{u}_i(\tilde{\bx},\bx',\delta_{1:i})&=F_{K_i}(x'_i+\delta_i|\bar{\bx}^i,\bar{y}^i_{1:i-1})-F_{K_i}(x'_i-\delta_i|\underline{\bx}^i,\underline{y}^i_{1:i-1})\\
&\leq F_{K_i}(x'_i+\delta_i|\tilde{\bx},x'_{1:i-1})-F_{K_i}(x'_i-\delta_i|\tilde{\bx},x'_{1:i-1})+0.5\tilde{K}\delta_i\\
&\leq \delta_i(2\bar{K}+0.5\tilde{K})\\
&\leq 2.5\delta_i\bar{K}
\end{align*}
where $\tilde{K}\leq \bar{K}$ is as in the proof of Lemma \ref{lem:dense}. (Note that $\bar{u}_i(\tilde{\bx},\bx',\delta_{1:i})-\underline{u}_i(\tilde{\bx},\bx',\delta_{1:i})$ is indeed strictly positive because $F_{K_i}(\cdot|\bx,y_{1:i-1},)$ is strictly increasing on $[0,1]$ for any $(\bx,\by)\in\setX^2$ and because $\delta_i>0$.)

This shows that for all $\bx\in B_{v_K(\delta)}(\tilde{\bx})$ and for all $\by\in B_{\delta_{1:i-1}}(\bx')$, we have
$$
\big[F_{K_i}(x'_i-\delta_i|\bx,\by),F_{K_i}(x'_i+\delta_i|\bx,\by)\big]\subseteq \big[\underline{u}_i(\tilde{\bx},\bx',\delta),\bar{u}_i(\tilde{\bx},\bx',\delta)\big],\quad\forall i\in 1:d
$$
and thus there exists   a closed hypercube $\bar{W}(\tilde{\bx},\bx',\delta)$ of side $\bar{S}_{\delta}=2.5\delta\bar{K}$ such that 
$$
F_K(\bx,B_{\delta_{1:i-1}}(\bx'))\subseteq  \bar{W}(\tilde{\bx},\bx',\delta),\quad \forall \bx\in B_{v_K(\delta)}(\tilde{\bx}).
$$
To conclude the proof of Lemma \ref{lem:upBound}, note that, because $v_K(\delta)\leq\delta_i$ for all $i\in 1:d$,
$$
F_K(\bx,B_{v_K(\delta)}(\bx'))\subseteq F_K(\bx,B_{\delta_{1:i-1}}(\bx')).
$$
\end{proof}

\begin{lemma}\label{lem:E_Set}
Consider the set-up of Lemma \ref{lem:convPhi} and, for $(p,a,k)\in \mathbb{N}_+^3$, let
\begin{align*}
E^{p}_{a,k}&=\Big\{\exists n\in \{ab^{k},\dots,(a+1)b^{k}-1\}:\ \bx^{n}\neq  \bx^{ab^{k}-1},\,\,\varphi(\bx^{ab^k-1})<\varphi^*\Big\}\\
&\cap \Big\{\forall n\in \{ab^{k},\dots,(a+1)b^{k}-1\}:\, \bx^{n}
\in(\setX_{\varphi(\bx^{ab^k-1})})_{2^{-p}}\Big\}.
\end{align*}
Then, for all for all $k\in\mathbb{N}$, there exists a $p^*_k\in\mathbb{N}$ such that $\Prob\big(\bigcap_{a\in\mathbb{N}}E^{p}_{a,k}\big)=0$ for all $p\geq p^*_k$.
\end{lemma}

\begin{proof}

Let $\epsilon>0$, $a\in\mathbb{N}$ and $l\in\mathbb{R}$ be such that  $l<\varphi^*$, and for $k\in\mathbb{N}$, let $E(k)=\{E(j,k)\}_{j=1}^{k^d}$ be the splitting of $\setX$ into closed hypercubes of volume $k^{-d}$.

Let $p'\in\mathbb{N}_+$, $\delta=2^{-p'}$   and  $P^l_{\epsilon,\delta}\subseteq E(\delta)$ be the smallest coverage of $(\setX_{l})_{\epsilon}$ by hypercubes in $E(\delta)$; that is, $|P^l_{\epsilon,\delta}|$ is the smallest value in $1:\delta^{-d}$ such that  $(\setX_{l})_{\epsilon}\subseteq \cup_{W\in P^l_{\epsilon,\delta}}$. Let $J^l_{\epsilon,\delta}\subseteq 1:\delta^{-d}$ be such that $j\in J^l_{\epsilon,\delta}$ if and only if $E(j,\delta)\in P^l_{\epsilon,\delta}$. We now bound $|J^l_{\epsilon,\delta}|$ following the same idea as in \citet{He2015}. 

By assumption, there exists a constant $\bar{M}<\infty$ independent of $l$ such that $M(\setX_{l})\leq \bar{M}$.
Hence, for any fixed $w>1$ there exists a $\epsilon^*\in (0,1)$ (independent of $l$) such that $\lambda_d\big((\setX_{l})_{\epsilon}\big)\leq w M(\setX_{l})\epsilon\leq  w \bar{M}\epsilon$ for all $\epsilon\in(0,\epsilon^*]$. Let $\epsilon=2^{-p}$, with $p\in\mathbb{N}$ such that $2^{-p}\leq 0.5\epsilon^*$, and take $\delta_{\epsilon}=2^{-p-1}$. Then, we have the inclusions
$(\setX_{l})_{\epsilon}\subseteq \cup_{W\in P^l_{\epsilon,\delta_{\epsilon}}} \subseteq (\setX_{l})_{2\epsilon}$ and therefore, since $2\epsilon\leq\epsilon^*$,
\begin{align}\label{eq:delta}
|J^l_{\epsilon,\delta_{\epsilon}}|\leq \frac{\lambda_d\big((\setX_{l})_{2\epsilon}\big)}{\lambda_d(E(j,\delta_{\epsilon}))}\leq\frac{w \bar{M} (2\epsilon)^{d}}{\delta_{\epsilon}^{d}}\leq  \bar{C}\delta_{\epsilon}^{-(d-1)},\quad \bar{C}:=w \bar{M} 2^d
\end{align}
where  the right-hand side is independent of $l$.

Next, for $j\in J^l_{\epsilon,\delta_{\epsilon}}$, let $\bar{\bx}^j$ be the center of $E(j,\delta_{\epsilon})$ and $W^l(j,\delta_{\epsilon})=\cup_{j'\in J^l_{\epsilon,\delta_{\epsilon}}} \bar{W}(\bar{\bx}^j,\bar{\bx}^{j'},\delta_{\epsilon})$, with $\bar{W}(\cdot,\cdot,\cdot)$  as in  Lemma \ref{lem:upBound}. Then, using this latter, a necessary condition to move at iteration $n+1$ of Algorithm \ref{alg:SAQMC} from a point $\bx^{n}\in E(j_{n},\delta_{\epsilon})$, with $j_{n}\in J^{l}_{\epsilon,\delta_{\epsilon}}$, to a point $\bx^{n+1}\neq \bx^{n}$ such that $\bx^{n+1}\in E(j_{n+1},\delta_{\epsilon})$ for a $j_{n+1}\in J^{l}_{\epsilon,\delta_{\epsilon}}$ is that $\bu_R^{n+1}\in W^{l}(j_{n},\delta_{\epsilon})$.

Let  $k^{\delta_{\epsilon}}$ be the largest  integer  such that (i) $b^{k}\leq \bar{S}_{\delta_{\epsilon}}^{-d}b^t$, with $\bar{S}_{\delta_{\epsilon}}=2.5\bar{K}\delta_{\epsilon}$, $\bar{K}<\infty$, as in Lemma \ref{lem:upBound}, and (ii)  $(k-t)/d$ is a positive integer (if necessary reduce $\epsilon$ to fulfil this last condition). 
Let $E'(\delta_{\epsilon})=\{E'(k,\delta_{\epsilon})\}_{k=1}^{b^{k^{\delta_{\epsilon}}-t}}$  be the partition of $\ui^d$ into hypercubes of volume $b^{t-k^{\delta_{\epsilon}}}$. Then, for all $j\in J^l_{\epsilon,\delta_{\epsilon}}$, $W^l(j,\delta_{\epsilon})$ is covered by at most $2^d|J^l_{\epsilon,\delta_{\epsilon}}|$  hypercubes of $E'(\delta_{\epsilon})$.

Let $\epsilon$ be small enough so that $k^{\delta_{\epsilon}}>t+dR$. Then, using the properties of $(t,s)_R$-sequences (see Section \ref{sub:(t-s)_R}), it is easily checked that,  for all $n\geq 0$,
\begin{align}\label{eq:proba}
\Prob\big(\bu_R^{n}\in E'(k,\delta_{\epsilon})\big)\leq b^{t-k^{\delta^{\epsilon}}+dR},\quad \forall k\in 1:b^{k^{\delta_{\epsilon}}-t}.
\end{align}
Thus, using \eqref{eq:delta}-\eqref{eq:proba} and the definition of $k^{\delta_{\epsilon}}$, we obtain, for all $j\in J^l_{\epsilon,\delta_{\epsilon}}$ and $n\geq 0$,
$$
\Prob\big(\bu_R^{n}\in W^l(j,\delta_{\epsilon})\big)\leq 2^d|J^l_{\epsilon,\delta_{\epsilon}}|b^tb^{t-k^{\delta^{\epsilon}}+dR}\leq C^*\delta_\epsilon,\quad C^*=2^d\bar{C}b^{t+1}(2.5\bar{K})^{d}b^{dR}.
$$
Consequently, using the definition of $\epsilon$ and $\delta_{\epsilon}$, and the fact that there exist at most $2^d$ values of $j\in J^l_{\epsilon,\delta_{\epsilon}}$ such that, for $n\in\mathbb{N}$, we have $\bx^{n}\in E(j,\delta_{\epsilon})$, we deduce that, for  a $p^*\in\mathbb{N}$ large enough (i.e. for  $\epsilon=2^{-p^*}$  small enough)
$$
\Prob\big(E^{p}_{a,k}|\, \varphi(\bx^{ab^k-1})=l\big)\leq b^{k} 2^dC^* 2^{-p-1},\quad\forall (a,k)\in\mathbb{N}^2,\quad \forall l<\varphi^*,\quad\forall p\geq p^*
$$
implying that, for $p\geq p^*$,
$$
\Prob\big(E^{p}_{a,k}\big)\leq b^{k} 2^dC^* 2^{-p-1},\quad\forall (a,k)\in\mathbb{N}^2.
$$
Finally,  because the uniform random numbers $\bz^n$'s in $\ui^s$ that enter into the definition of $(t,s)_R$-sequences are IID, this shows that
$$
\Prob\big(\cap_{j=a}^{a+m}E_{j,k}^p\big)\leq (b^{k}2^dC^*2^{-p-1})^m,\quad\forall(a,m,k)\in\mathbb{N}^3,\quad \forall p\geq p^*.
$$ 
To conclude the proof, for $k\in\mathbb{N}$ let $\rho_k\in (0,1)$ and $p^*_k\geq p^*$ be such that 
$$
b^{k}2^dC^*2^{-p-1}\leq\rho_k,\quad\forall p>p^*_k.
$$
Then, $\Prob\big(\cap_{a\in\mathbb{N}} E^p_{a,k})=0$ for all $p\geq p^*_k$, as required.
\end{proof}

\begin{lemma}\label{lem:part_2}
Consider the set-up of Lemma \ref{lem:convPhi}. For $k\in\mathbb{N}$, let $\tilde{E}(dk)=\{\tilde{E}(j,k)\}_{j=1}^{b^{dk}}$ be the partition of $\ui^d$ into hypercubes of volume $b^{-dk}$. Let $k^{R} \in (dR+t):(dR+t+d)$ be the smallest integer $k$ such  $(k-t)/d$ is an integer and such that $(k-t)/d\geq R$ and, for $m\in\mathbb{N}$, let $I_m=\{mb^{k^R},\dots,(m+1)b^{k^R}-1\}$.
 Then, for any $\delta\in(0,\bar{\delta}_K]$ verifying $k_\delta> t+d+dR$ (with $\bar{\delta}_K$ and $k_\delta$ as in Lemma \ref{lem:dense}), there exists a $p(\delta)>0$ such that
$$
\Prob\big(\exists n\in I_m:\,\, \bu_R^{n}\in \tilde{E}(j, k_{\delta}-t_{\delta,d})\big)\geq  p(\delta),\quad \forall j\in 1: b^{ k_{\delta}-t_{\delta,d}},\quad\forall m\in\mathbb{N}
$$
where $t_{\delta,d}\in t:(t+d)$ is such that $(k_{\delta}-t_{\delta,d})/d\in\mathbb{N}$.
\end{lemma}
\begin{proof}

Let $m\in\mathbb{N}$ and  note that, by the properties of $(t,s)_R$-sequence,  the point set $\{\bu_{\infty}^{n}\}_{n\in I_m}$  is a $(t,{k^R},d)$-net in base $b$. Thus, for all $j\in 1:b^{k^R-t}$, this point set contains $b^t$ points in $\tilde{E}(j, k^{R}-t)$  and,  consequently, for all $j\in 1:b^{dR}$, it contains    $b^tb^{k^R-t-dR}=b^{k^R-dR}\geq 1$ points in $\tilde{E}(j,dR)$. This implies that, for all $j\in 1:b^{dR}$, the point set  $\{\bu_R^{n}\}_{n\in I_m}$ contains  $b^{k^R-dR}\geq 1$ points in  $\tilde{E}(j,dR)$ where, for all $ n\in  I_{m_i}$, $\bu_R^{n}$ is uniformly distributed in $\tilde{E}(j_{n},dR)$ for a $j_n\in 1:b^{dR}$. 

In addition, it is easily checked that each hypercube of the set $\tilde{E}(dR)$ contains 
$$
b^{k_{\delta}-t_{\delta,d}-dR}\geq b^{k_{\delta}-t-d-dR}>  1
$$
hypercubes of the set $\tilde{E}(k_{\delta}-t_{\delta,d})$, where $k_{\delta}$ and $t_{\delta,d}$ are as in the statement of the lemma. Note that the last inequality holds because $\delta$ is chosen so that $k_{\delta}> t+d+dR$. Consequently, 
$$
\Prob\big(\exists n\in I_m:\,\, \bu_R^n\in \tilde{E}(j, k_{\delta}-t_{\delta,d})\big)\geq  p(\delta):=b^{dR+t_{\delta,d}-k_{\delta}}>0,\quad \forall j\in 1: b^{ k_{\delta}-t_{\delta,d}}
$$
and the proof is complete.
\end{proof}

\emph{Proof of Lemma \ref{lem:convPhi}:} To prove the lemma we need to introduce some additional notation. Let $\Omega=\ui^{\mathbb{N}}$, $\mathcal{B}(\ui)$ be the Borel $\sigma$-algebra on $\ui$, $\F= \mathcal{B}(\ui)^{\otimes\mathbb{N}}$ and $\P$ be the probability measure on $(\Omega,\F)$ defined by
$$
\P(A)=\prod_{i\in\mathbb{N}}\lambda_1(A_i),\quad (A_1,\dots,A_i\dots)\in \mathcal{B}(\ui)^{\otimes\mathbb{N}}.
$$
Next,  for  $\omega\in\Omega$, we denote by $\big(\bbu_R^n(\omega)\big)_{n\geq 0}$ the sequence  of points in $\ui^{d}$ defined, for all $n\geq 0$, by (using the convention that empty sums are null),
$$
\bbu_R^n(\omega)=\big(U_{R,1}^n(\omega),\dots,U_{R,d}^n(\omega)),\quad U_{R,i}^n(\omega)=\sum_{k=1}^{R}a_{ki}^nb^{-k}+b^{-R}\omega_{n d+i},\quad i\in 1:s.
$$
Note that, under $\P$, $\big(\bbu_R^n\big)_{n\geq 0}$ is a $(t,d)_R$-sequence in base $b$. Finally, for $\omega\in\Omega$, we denote by $\big(\bx^n_\omega\big)_{n\geq 0}$ the sequence of points in $\setX$ generated by Algorithm \ref{alg:SAQMC} when the sequence $\big(\bbu_R^n(\omega)\big)_{n\geq 0}$ is used as input.

Under the assumptions of the lemma there exists a set $\Omega_1\in\F$  such that $\P(\Omega_1)=1$ and
$$
\exists\bar{\varphi}_\omega\in\mathbb{R}\text{ such that }\lim_{n\rightarrow\infty}\varphi\big(\bx^n_\omega\big)=\bar{\varphi}_\omega,\quad \forall\omega\in\Omega_1.
$$

Let $\omega\in\Omega_1$. Since $\varphi$ is continuous, for any $\epsilon>0$  there exists a $N_{\omega, \epsilon}\in\mathbb{N}$ such that $\bx^n_\omega\in(\setX_{\bar{\varphi}_\omega})_{\epsilon}$ for all $n\geq N_{\omega,\epsilon}$, where we recall that  $(\setX_{\bar{\varphi}_\omega})_{\epsilon}=\{\bx\in\setX:\exists\bx'\in\setX_{\bar{\varphi}_\omega}\text{ such that } \|\bx-\bx'\|_{\infty}\leq\epsilon\}$. In addition, because $\varphi$ is continuous and  $\setX$ is compact, there exists an integer $p_{\omega, \epsilon}\in\mathbb{N}$ such that we have both $\lim_{\epsilon\cvz}p_{\omega,\epsilon}=\infty$ and
\begin{align}\label{eq:inclusion}
(\setX_{\bar{\varphi}_\omega})_{\epsilon}\subseteq (\setX_{\varphi(x')})_{2^{-p_{\omega,\epsilon}}},\quad\forall x'\in(\setX_{\bar{\varphi}_\omega})_{\epsilon}.
\end{align}
Next, let $\bx^*\in\setX$ be such that $\varphi(\bx^*)=\varphi^*$, $k^R \in (dR+t):(dR+t+d)$ be as in Lemma \ref{lem:part_2} and, for $(p,a,k)\in\mathbb{N}_+^3$, let 
\begin{align*}
\tilde{E}^{p}_{a,k}&=\Big\{\omega\in\Omega:\,\exists n\in \{ab^{k},\dots,(a+1)b^{k}-1\}:\ \bx_{\omega}^{n}\neq  \bx^{ab^{k}-1}_{\omega},\,\varphi(\bx_{\omega}^{ab^k-1})<\varphi^*\Big\}\\
&\cap \Big\{\omega\in\Omega:\,\forall n\in \{ab^{k},\dots,(a+1)b^{k}-1\}:\bx^{n}_{\omega}\in\big(\setX_{\varphi(\bx^{ab^k-1}_{\omega})}\big)_{2^{-p}}\Big\}.
\end{align*}
Then, by Lemma \ref{lem:E_Set}, there exists a $p^*\in\mathbb{N}$ such that $\P\big(\cap_{a\in\mathbb{N}} \tilde{E}^p_{a,k^R}\big)=0$ for all $p\geq p^*$, and thus the set $\tilde{\Omega}_2=\cap_{p\geq p^*}\big(\setX\setminus \cap_{a\in\mathbb{N}} \tilde{E}^p_{a,k^R}\big)$ verifies $\P(\tilde{\Omega}_2)=1$. Let $\Omega_2=\Omega_1\cap \tilde{\Omega}_2$ so that $\P(\Omega_2)=1$. 

For $\omega\in\Omega_2$ let  $\epsilon_\omega>0$ be small enough so that, for any $\epsilon\in(0,\epsilon_\omega]$, we can take  $p_{\omega,\epsilon}\geq p^*$ in \eqref{eq:inclusion}. Then,  for any $\omega\in\Omega_2$ such that $\bar{\varphi}_\omega<\varphi^*$,  there exists a subsequence $(m_i)_{i\geq 1}$ of $(m)_{m\geq 1}$  such that, for all $i\geq 1$, either 
$$
\bx^{n}_\omega=\bx^{m_ib^{k^R}-1}_\omega,\quad \forall n\in I_{m_i}:=\big\{m_i b^{k^R},\dots, (m_i+1)b^{k^R}-1\big\}
$$
or 
\begin{align}\label{eq:Crit2}
\exists n\in I_{m_i}\text{ such that }\bx^{n}_\omega
\not\in\Big(\setX_{\varphi(\bx^{m_ib^{k^R}-1}_\omega)}\Big)_{2^{-p_{\omega,\epsilon}}}.
\end{align}

Assume first that there exist  infinitely many $i\in\mathbb{N}$ such that \eqref{eq:Crit2} holds. Then,  by \eqref{eq:inclusion}, this leads to a contradiction with  the fact  that $\omega\in\Omega_2\subseteq\Omega_1$. Therefore, for any $\omega\in\Omega_2$ such that $\bar{\varphi}_{\omega}<\varphi^*$ there exists a subsequence $(m_i)_{i\geq 1}$ of $(m)_{m\geq 1}$  such that, for a $i^*$ large enough,
\begin{align}\label{eq:mbar}
\bx_\omega^{n}=\bx_\omega^{m_ib^{k^R}-1},\quad \forall  n\in I_{m_i},\quad\forall i\geq i^*.
\end{align}
Let $
\tilde{\Omega}_2=\{\omega\in\Omega_2:\,\bar{\varphi}_{\omega}<\varphi^*\}\subseteq\Omega_2$ . Then, to conclude the proof, it remains to show that $\P(\tilde{\Omega}_2)=0$. We prove this result by contradiction and thus, from henceforth, we assume $\P(\tilde{\Omega}_2)>0$.

To this end, let $\bx^*\in\setX$ be such that $\varphi(\bx^*)=\varphi^*$, $\bx\in\setX$ and $\delta\in(0,\bar{\delta}_K]$, with $\bar{\delta}_K$ as in Lemma \ref{lem:dense}. Then, using this latter,  a sufficient condition to have $F_K^{-1}(\bx,\bbu_R^{n}(\omega))\in B_{\delta}(\bx^*)$, $n\geq 1$,  is that $\bbu_R^{n}(\omega)\in \underline{W}(\bx,\bx^*,\delta)$, with $\underline{W}(\cdot,\cdot,\cdot)$   as in Lemma \ref{lem:dense}. From the proof of this latter  we know that the hypercube $\underline{W}(\bx,\bx^*,\delta)$ contains at least one hypercube of the set $\tilde{E}(k_{\delta}-t_{\delta,d})$, where $t_{\delta,d}\in t:(t+d)$ is such that $(k_{\delta}-t_{\delta,d})/d\in\mathbb{N}$ and, for $k\in\mathbb{N}$, $\tilde{E}(dk)$ is as in Lemma \ref{lem:part_2}. Hence, by this latter, for any $\delta\in(0,\delta^*]$, with $\delta^*$ such that $k_{\delta^*}>t+d+dR$ (where, for $\delta>0$,  $k_\delta$ is defined in Lemma \ref{lem:dense}), there exists a $p(\delta)>0$  such that
$$
\P\Big(\omega\in\Omega:\, \exists n\in I_m,\, F_K^{-1}\big(\bx, \bbu_R^n(\omega)\big)\in B_{\delta}(\bx^*)\Big)\geq p(\delta),\quad\forall (\bx,m)\in\setX\times\mathbb{N}
$$
and thus, using  \eqref{eq:mbar}, it is easily checked that, for any $\delta\in(0,\delta^*]$,
$$
\P_2\Big(\omega \in\tilde{\Omega}_2: F_K^{-1}\big(\bx_{\omega}^{n-1}, \bbu_R^n(\omega)\big)\in B_{\delta}(\bx^*)\text{ for infinitly many }n\in\mathbb{N}\Big)=1
$$
where $\P_2$ denotes the restriction of $\P$ on $\tilde{\Omega}_2$ (recall that we assume $\P(\tilde{\Omega}_2)>0$).

For $\delta>0$, let
$$
\Omega'_{\delta}=\Big\{\omega \in\tilde{\Omega}_2:\,\, F_K^{-1}(\bx_{\omega}^{n-1}, \bbu_R^n(\omega))\in B_{\delta}(\bx^*)\text{ for infinitly many }n\in\mathbb{N}\Big\}
$$
and let $\tilde{p}^*\in\mathbb{N}$ be such that $2^{-\tilde{p}^*}\leq\delta^*$. Then, the set $\Omega'=\cap_{\tilde{p}\geq   \tilde{p}^*}\Omega'_{2^{-\tilde{p}}}$ verifies $\P_2(\Omega')=1$.

To conclude the proof let $\omega\in \Omega'$. Then, because $\varphi$ is continuous and $\bar{\varphi}_\omega<\varphi^*$, there exists a $\tilde{\delta}_{\bar{\varphi}_\omega}>0$ so that $\varphi(\bx)>\bar{\varphi}$ for all $\bx\in B_{\tilde{\delta}_{\bar{\varphi}_\omega}}(\bx^*)$. Let  $\delta_{\bar{\varphi}_\omega}:=2^{-\tilde{p}_{\omega,\epsilon}}\geq \tilde{\delta}_{\bar{\varphi}_\omega}\wedge\bar{\delta}_K$ for an integer  $\tilde{p}_{\omega,\epsilon}\geq \tilde{p}^*$. Next, take $\epsilon$ small enough  so that  we have both $B_{\delta_{\bar{\varphi}_\omega}}(\bx^*)\cap (\setX_{\bar{\varphi}_\omega})_{\epsilon}=\emptyset$ and $\varphi(\bx)\geq \varphi(\bx')$ for all $(\bx,\bx')\in B_{\delta_{\bar{\varphi}_\omega}}(\bx^*)\times (\setX_{\bar{\varphi}_\omega})_{\epsilon}$. 

Using above computations, the set $ B_{\tilde{\delta}_{\bar{\varphi}_\omega}}(\bx^*)$ is visited infinitely many time and thus $\varphi(\bx_\omega^n)>\bar{\varphi}_\omega$ for infinitely many $n\in\mathbb{N}$, contradicting the fact that $\varphi(\bx_\omega^n)\rightarrow\bar{\varphi}_\omega$ as $n\rightarrow\infty$. Hence, the set $\Omega'$ is empty. On the other hand, as shown above,  under the assumption   $\P(\tilde{\Omega}_2)>0$  we have $\P_2(\Omega')=1$ and, consequently, $\Omega'\neq\emptyset$. Therefore, we must have $\P(\tilde{\Omega}_2)=0$ and the proof is complete.

\section{Proof of Theorem \ref{thm:conv_Univ}\label{p-thm:conv_Univ}}

Using Lemmas \ref{lem:Tn} and \ref{lem:ln}, we know that $\varphi(x^n)\rightarrow\bar{\varphi}\in\mathbb{R}$ and thus it remains to show that $\bar{\varphi}=\varphi^*$.

Assume that $\bar{\varphi}\neq\varphi^*$ and, for $\epsilon=2^{-p}$, $p\in\mathbb{N}_+$, let $N_\epsilon\in\mathbb{N}$, $p_\epsilon$ and $\delta_{\bar{\varphi}}>0$ be as in the proof of Lemma \ref{lem:convPhi} (with the dependence of $N_\epsilon$, $p_\epsilon$ and of  $\delta_{\bar{\varphi}}$ on $\omega\in\Omega$ suppressed in the notation for obvious reasons).

Let $x^*\in\setX$ be a global maximizer of $\varphi$ and  $n=a_nb^{k_{\delta_{\bar{\varphi}}}}-1$ with $a_n\in\mathbb{N}$  such that $n>N_{\epsilon}$. For $k\in\mathbb{N}$, let $E(k)=\{E(j,k)\}_{j=1}^{k}$ be the splitting of $[0,1]$ into closed hypercubes of volume $k^{-1}$. Then, by Lemma \ref{lem:upBound},  a necessary condition to have a move at iteration $n'+1\geq 1$ of Algorithm \ref{alg:SAQMC} from $x^{n'}\in (\setX_{\bar{\varphi}})_{\epsilon}$ to $x^{n'+1}\neq x^{n'}$, $x^{n'+1}\in (\setX_{\bar{\varphi}})_{\epsilon}$ is that 
$$
u_{\infty}^{n'}\in\bar{W}(\epsilon):=\bigcup_{j,j'\in J^{\bar{\varphi}}_{\epsilon, \epsilon/2}} \bar{W}(\bar{x}^j,\bar{x}^{j'},\epsilon/2)
$$
where, for $j\in 1:(\epsilon/2)^{-d}$, $\bar{x}^j$ denotes the center of $E(j,\epsilon/2)$, $J^{\bar{\varphi}}_{\epsilon, \epsilon/2}$ is as in the proof of Lemma \ref{lem:E_Set} and $\bar{W}(\cdot,\cdot,\cdot)$ is as in Lemma \ref{lem:upBound}.  Note that, using \eqref{eq:delta} with $d=1$, $|J^{\bar{\varphi}}_{\epsilon, \epsilon/2}|\leq C^*$ for a constant $C^*<\infty$ (independent of $\epsilon$).

Let $b^{k^{\delta_{\epsilon}}}$ be the largest integer $k\geq t$ such that $b^{t-k}\geq \bar{S}_{\epsilon/2}^{d}$, with $\bar{S}_{\epsilon/2}$ as in Lemma \ref{lem:upBound}, and let $\epsilon$ be small enough so that $b^{k^{\delta_{\epsilon}}}>2^dC^*b^t$. The point set $\{u_{\infty}^{n'}\}_{n'=a_nb^{k^{\delta_{\epsilon}}}}^{
(a_n+1)b^{k^{\delta_{\epsilon}}}-1}$ is a $(t, k^{\delta_{\epsilon}},d)$-net in base $b$ and thus the set $\bar{W}(\epsilon)$ contains at most $2^dC^* b^t$ points of this points set. Hence, if for $n>N_{\epsilon}$ only moves inside the set $(\setX_{\bar{\varphi}})_{\epsilon}$ occur, then, for a $\tilde{n}\in a_nb^{k^{\delta_{\epsilon}}}:\big((a_n+1)b^{k^{\delta_{\epsilon}}}-\eta_{\epsilon}-1)\big)$, the point set $\{x^{n'}\}_{n'=\tilde{n}}^{\tilde{n}+\eta_{\epsilon}}$ is such that $x^{n'}=x^{\tilde{n}}$ for all $n\in \tilde{n}:(\tilde{n}+\eta_{\epsilon})$, where $\eta_{\epsilon}\geq \lfloor \frac{b^{k^{\delta_{\epsilon}}}}{2^dC^*{}^2b^t}\rfloor$; note that $\eta_{\epsilon}\rightarrow\infty$ as $\epsilon\cvz$. 

Let $k^{\epsilon}_0$ be the largest integer which verifies $\eta_{\epsilon}\geq 2b^{k^{\epsilon}_0}$ so that $\{u_{\infty}^n\}_{n=\tilde{n}}^{\tilde{n}+\eta_{\epsilon}}$ contains at least one $(t,k^{\epsilon}_0,d)$-net in base $b$.  Note that $k^{\epsilon}_0\rightarrow\infty$ as $\epsilon\cvz$, and let $\epsilon$ be small enough so that $k^{\epsilon}_0\geq k_{\delta_{\bar{\varphi}}}$, with $k_{\delta}$ as in Lemma \ref{lem:dense}. Then, by Lemma \ref{lem:dense}, there exists at least one $n^*\in(\tilde{n}+1):(\tilde{n}+\eta_{\epsilon})$ such that $\tilde{y}^{n^*}:=F_K^{-1}(x^{\tilde{n}},u_{\infty}^{n^*})\in B_{\delta_{\bar{\varphi}}}(x^*)$. Since,  by the definition of $\delta_{\bar{\varphi}}$, for all $(x,x')\in B_{\delta_{\bar{\varphi}}}(x^*)\times (\setX_{\bar{\varphi}})_{\epsilon}$, and for $\epsilon$ small enough, we have $\varphi(x)>\varphi(x')$, it follows that $\varphi(\tilde{y}^{n^*})>\varphi(x^{\tilde{n}})$. Hence, there exists at least one $n\in\tilde{n}:(\tilde{n}+\eta_{\epsilon})$ such that $x^n\neq x^{\tilde{n}}$, which contradicts the fact that $x^n=\tilde{x}$ for all $n\in \tilde{n}:(\tilde{n}+\eta_{\epsilon})$. This shows that $\bar{\varphi}$ is indeed the global maximum of $\varphi$.

\newpage

\section{Additional figures for the example of Section \ref{sub:toy}\label{app:fig}}

\begin{figure}[h]
\centering
\begin{subfigure}{0.28\textwidth}
\centering
\includegraphics[trim=2cm 0 0 0, scale=0.23]{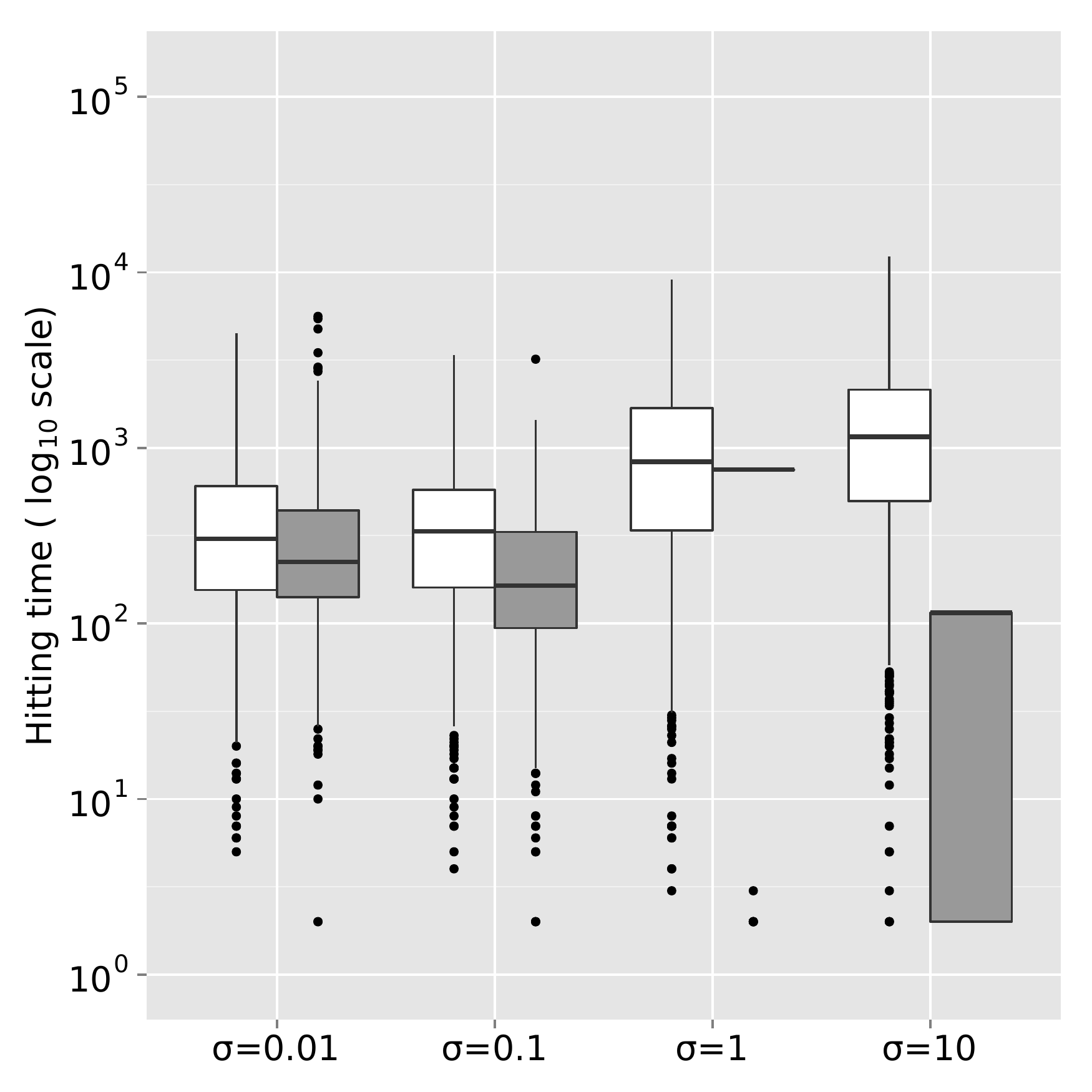}
\caption{\label{Cauchy_QMC1}}
\end{subfigure}
\hspace{1.2cm}
\begin{subfigure}{0.28\textwidth}
\includegraphics[trim=2cm 0 0 0, scale=0.23]{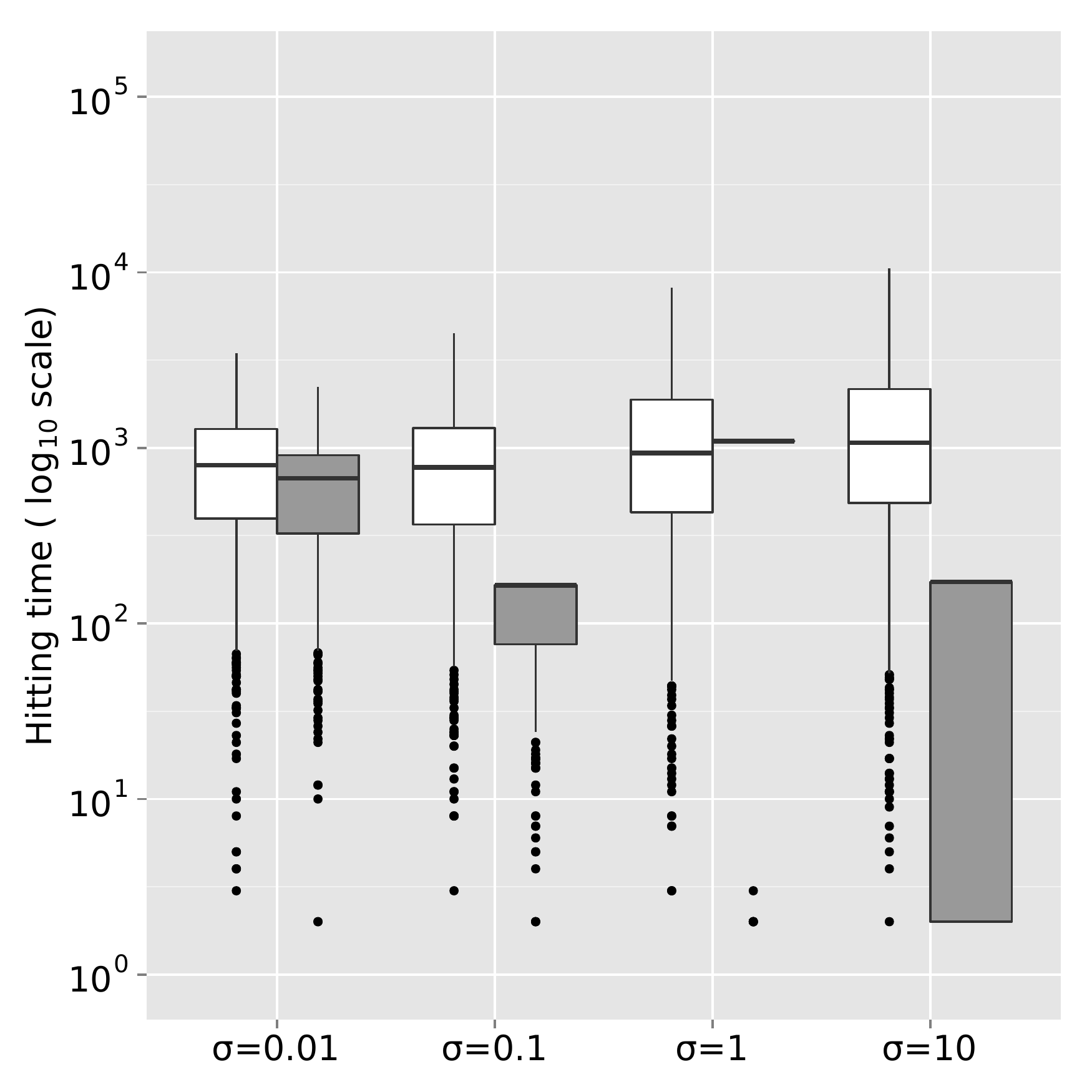}
\caption{\label{Cauchy_QMC3}}
\end{subfigure}

\begin{subfigure}{0.28\textwidth}
\centering
\includegraphics[trim=2cm 0 0 0, scale=0.23]{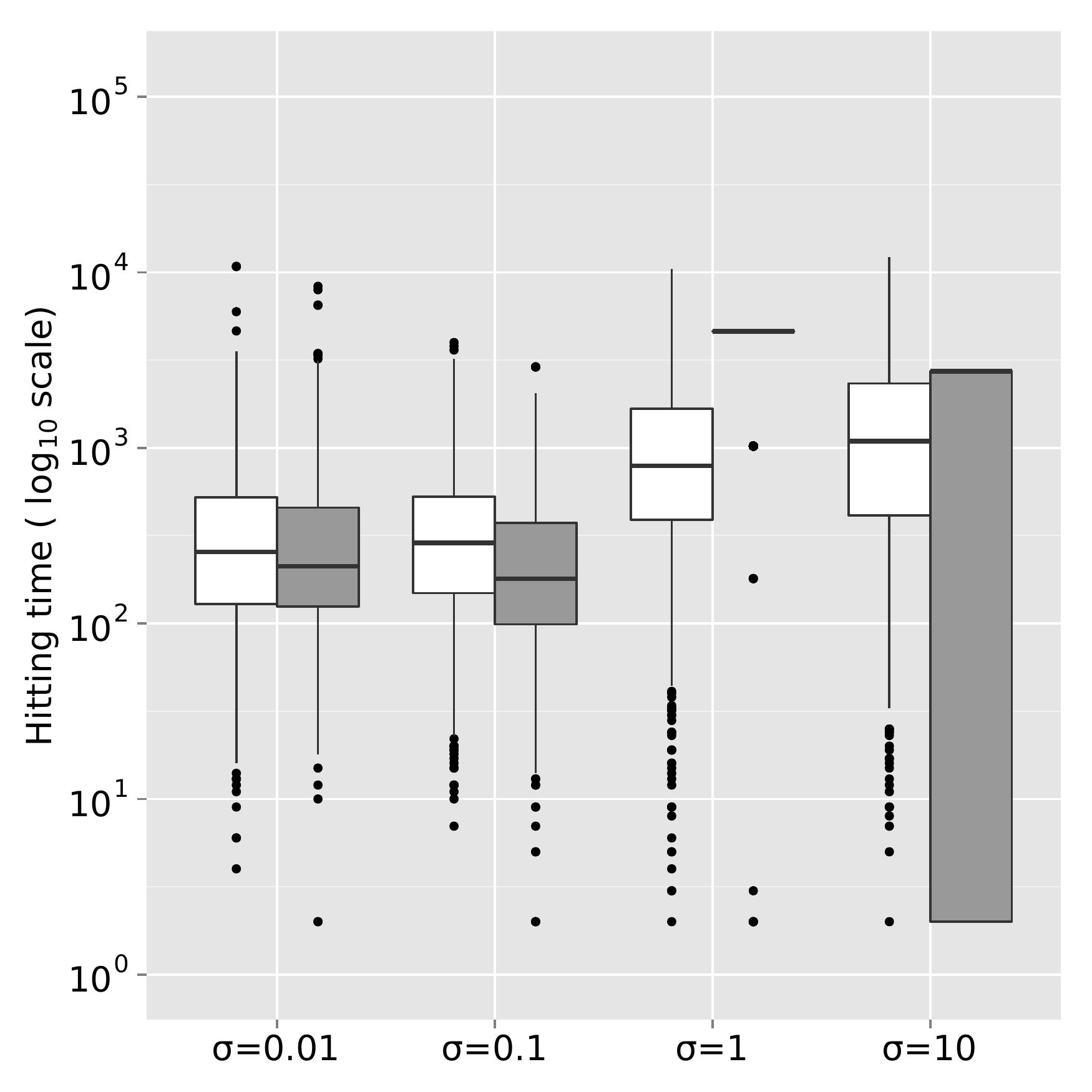}
\caption{\label{Cauchy_Lin1}}
\end{subfigure}
\hspace{1.2cm}
\begin{subfigure}{0.28\textwidth}
\includegraphics[trim=2cm 0 0 0, scale=0.235]{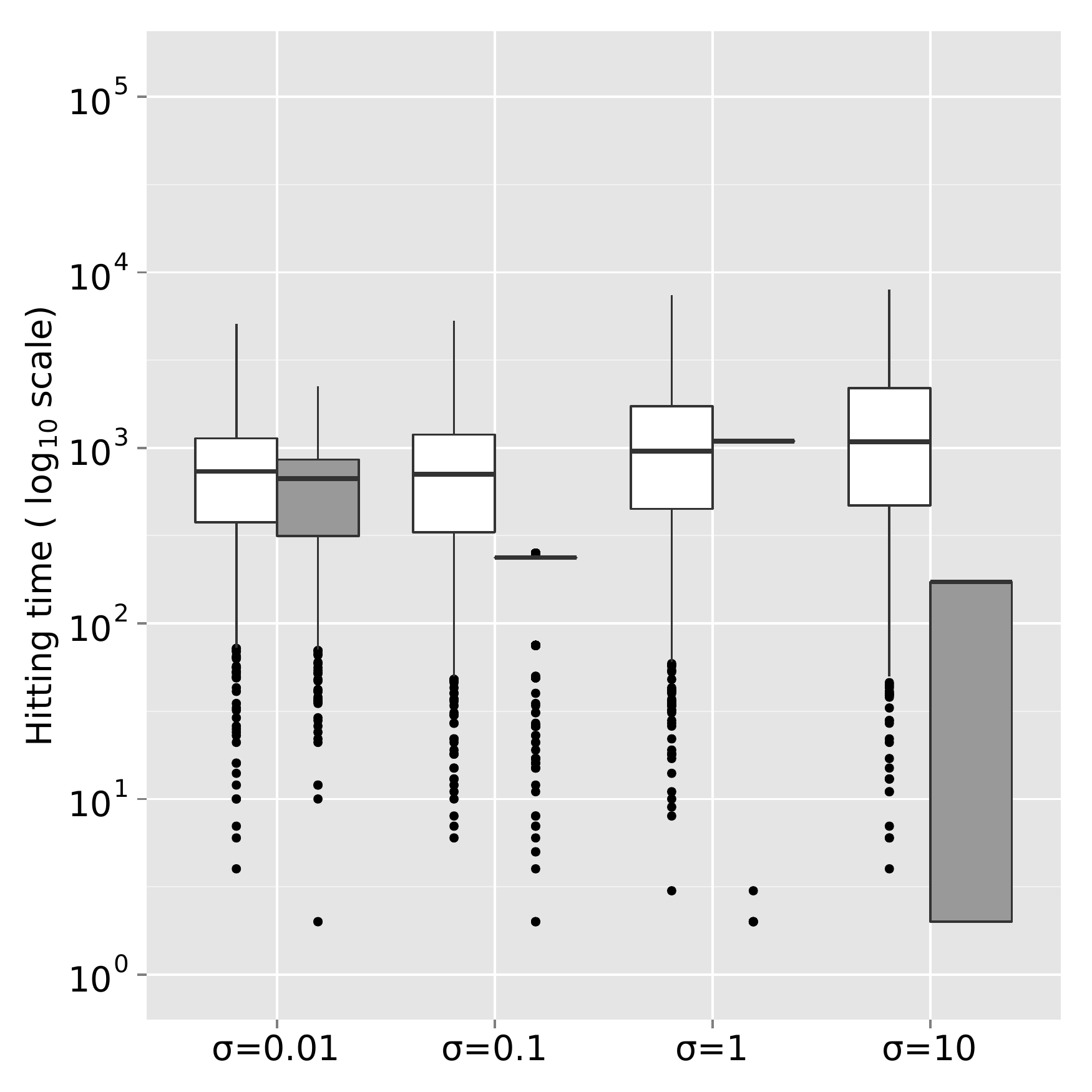}
\caption{\label{Cauchy_Lin3}}
\end{subfigure}

\caption{Minimization of $\tilde{\varphi}_1$ defined by \eqref{eq:toy} for 1\,000  starting values sampled independently and uniformly on $\setX_1$. Results are presented for the Cauchy kernel $(K_n^{(2)})_{n\geq 1}$  with $(T_n^{(1)})_{n\geq 1}$ (top plots) and for $(T_n^{(2)})_{n\geq 1}$. For $m=1,2$, simulations are done for the smallest (left plots) and the highest (right) value of $T_n^{(m)}$  given in \eqref{num:Temp2}. The plots show the minimum number of iterations needed for SA (white boxes) and QMC-SA to find a $\bx\in \setX_1$ such that $\tilde{\varphi}_1(\bx)<10^{-5}$. For each starting value, the Monte Carlo algorithm is run only once and the QMC-SA algorithm is based on the Sobol' sequence.\label{fig:Robert_Cauchy}}
\end{figure}
\newpage

\begin{figure}[h]
\centering
\begin{subfigure}{0.28\textwidth}
\centering
\includegraphics[ scale=0.25]{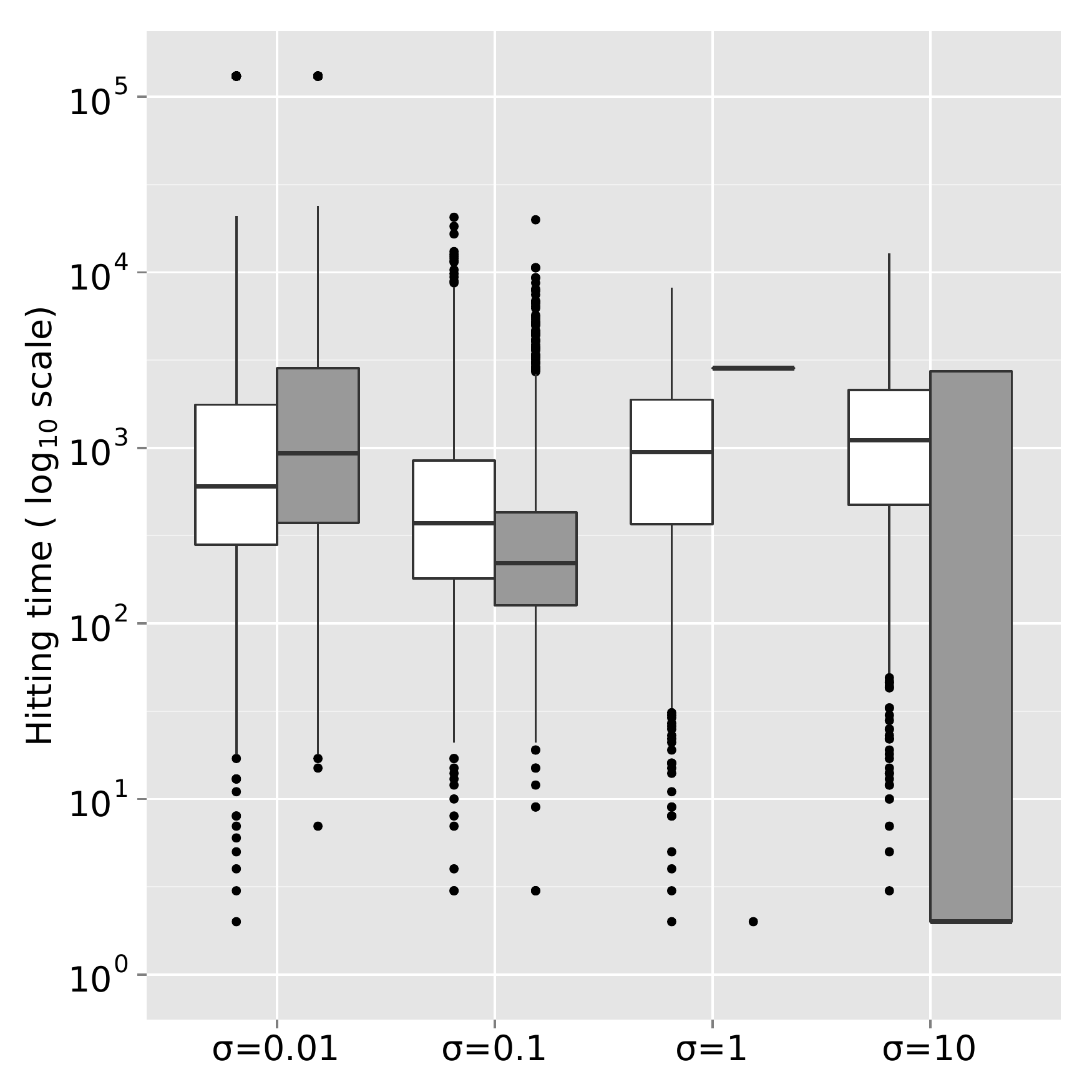}
\caption{\label{Gaussian_QMC1}}
\end{subfigure}
\hspace{1.2cm}
\begin{subfigure}{0.28\textwidth}
\includegraphics[scale=0.25]{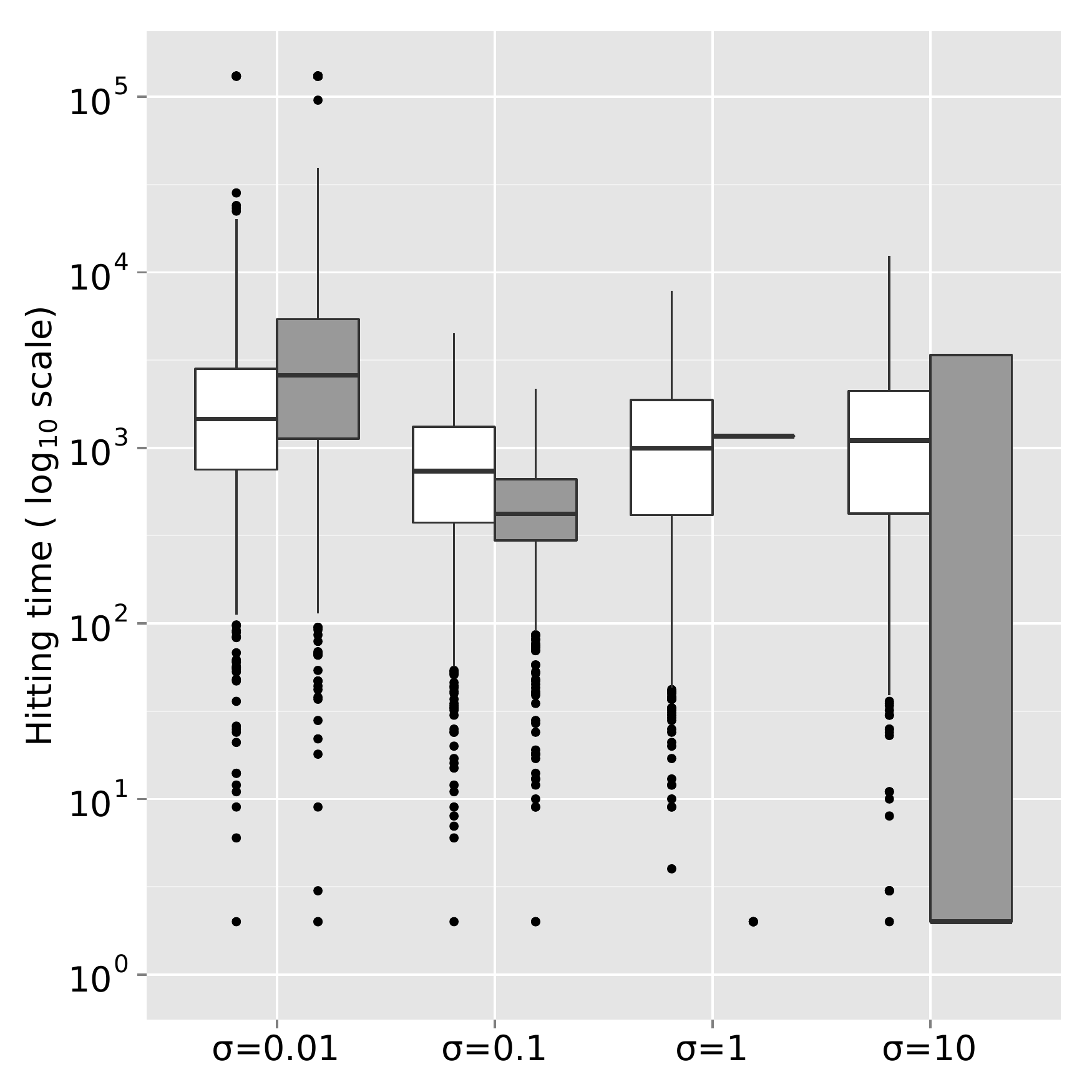}
\caption{\label{Gaussian_QMC3}}
\end{subfigure}

\begin{subfigure}{0.28\textwidth}
\centering
\includegraphics[ scale=0.25]{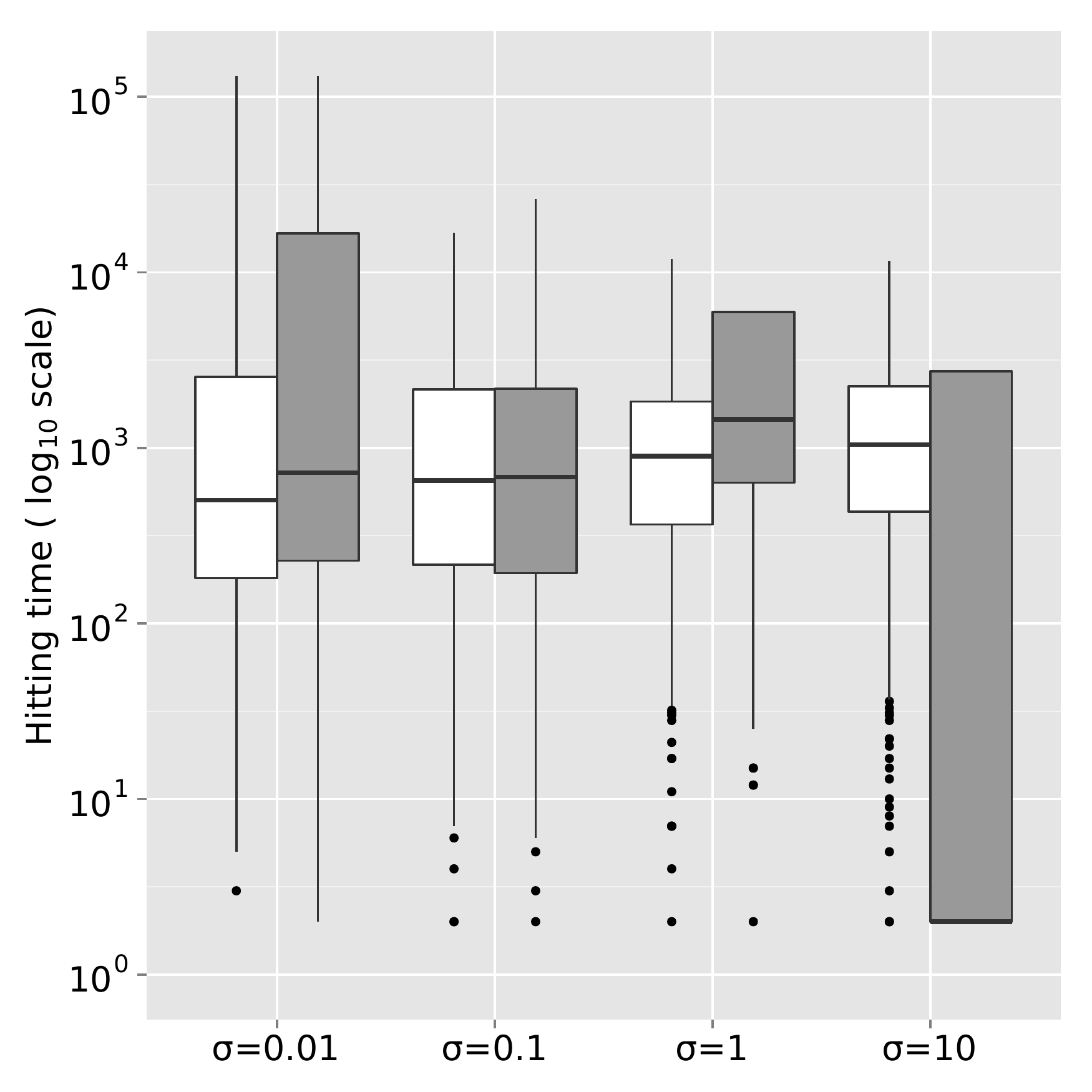}
\caption{\label{Gaussian_Log1}}
\end{subfigure}
\hspace{1.2cm}
\begin{subfigure}{0.28\textwidth}
\includegraphics[ scale=0.25]{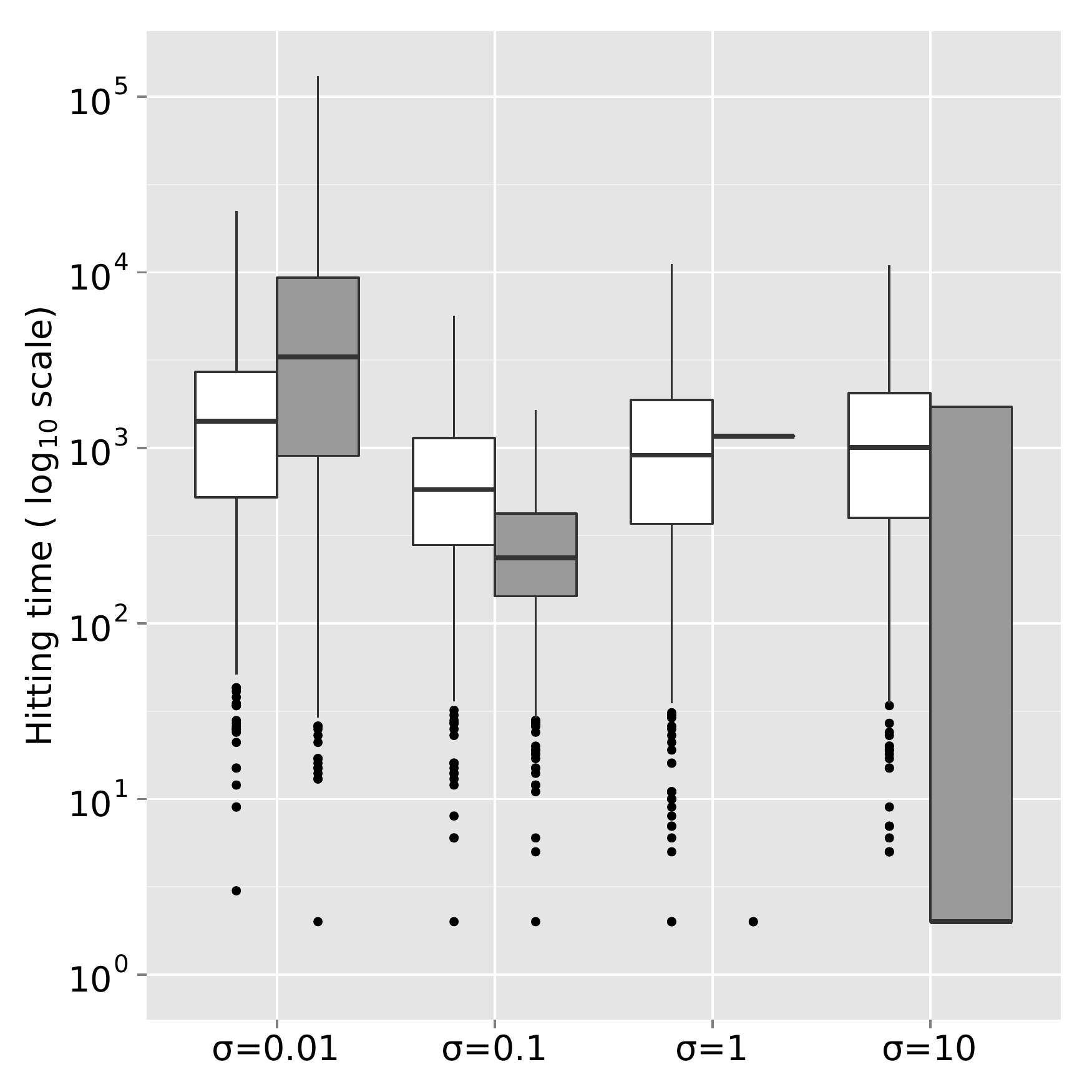}
\caption{\label{Gaussian_Log3}}
\end{subfigure}
\caption{Minimization of $\tilde{\varphi}_1$ defined by \eqref{eq:toy} for 1\,000  starting values sampled independently and uniformly on $\setX_1$. Results are presented for the Gaussian kernel $(K_n^{(3)})_{n\geq 1}$  with $(T_n^{(1)})_{n\geq 1}$ (top plots) and for $(T_n^{(3)})_{n\geq 1}$. For $m=\{1,3\}$, simulations are done for the smallest (left plots) and the highest (right) value of $T_n^{(m)}$  given in \eqref{num:Temp2}. The plots show the minimum number of iterations needed for SA (white boxes) and QMC-SA to find a $\bx\in \setX_1$ such that $\tilde{\varphi}_1(\bx)<10^{-5}$. For each starting value, the Monte Carlo algorithm is run only once and the QMC-SA algorithm is based on the Sobol' sequence.\label{fig:Robert_Gaussian}}
\end{figure}

\newpage
\bibliographystyle{apalike}
\bibliography{complete}

\end{document}